\documentclass[twocolumn, aps, superscriptaddress]{revtex4}

\newcommand{\bra}[1]{\langle #1|}
\newcommand{\ket}[1]{|#1\rangle}

\usepackage{amsmath,amsthm,amsfonts,amssymb}
\usepackage{fontenc}
\usepackage[all]{xy}
\usepackage{graphicx,subfigure,multirow}
\usepackage{tikz}
\usetikzlibrary{shapes, arrows, shadows}
\usepackage[latin1]{inputenc}
\usepackage[scaled]{helvet}
\usepackage[toc, page]{appendix}
\usepackage{graphicx}				
\usepackage{amssymb}
\usepackage{dsfont}
\newtheorem{definition}{Definition}
\newtheorem{theorem}{Theorem}
\newtheorem{lemma}{Lemma}

\usepackage{tikz}
\pgfdeclarelayer{nodelayer}
\pgfdeclarelayer{edgelayer}
\pgfsetlayers{edgelayer,nodelayer,main}
\tikzstyle{none}=[inner sep=0pt]
\tikzstyle{new}=[inner sep=2pt]
\usetikzlibrary{shapes, arrows, shadows}
\usetikzlibrary{circuits.ee.IEC}
\usetikzlibrary{decorations.pathmorphing}
\usetikzlibrary{shapes.geometric}

\tikzstyle{env}=[copoint,regular polygon rotate=0,minimum width=0.2cm, fill=black]

\tikzstyle{probs}=[shape=semicircle,fill=white,draw=black,shape border rotate=180,minimum width=1.2cm]

\usetikzlibrary{circuits.ee.IEC}
\usetikzlibrary{decorations.pathmorphing}
\usetikzlibrary{shapes.geometric}

\tikzstyle{wavy}=[decorate,decoration={snake, segment length=1mm, amplitude=0.3mm}]
\tikzstyle{mopoint}=[shape=semicircle, fill=white,draw=black,shape border rotate=180,scale =0.75]
\tikzstyle{mocopoint}=[shape=semicircle, fill=white,draw=black,minimum width = 0.9cm, scale =0.75, xscale=0.7]
\tikzstyle{cpoint}=[shape=semicircle, fill=white,draw=black,minimum width = 0.9cm, scale =0.75, xscale=1, yscale=0.7, shape border rotate = 90,font=\fontsize{14}{16}\selectfont]
\tikzstyle{cocpoint}=[shape=semicircle, fill=white,draw=black,minimum width = 0.9cm, scale =0.75, xscale=1, yscale=0.7, shape border rotate = 270,font=\fontsize{14}{16}\selectfont]

\tikzstyle{every picture}=[baseline=-0.25em,scale=0.5]
\tikzstyle{dotpic}=[] \tikzstyle{diredges}=[every to/.style={diredge}]
\tikzstyle{math matrix}=[matrix of math nodes,left delimiter=(,right delimiter=),inner sep=2pt,column sep=1em,row sep=0.5em,nodes={inner sep=0pt},text height=1.5ex, text depth=0.25ex]

\tikzstyle{inline text}=[text height=1.5ex, text depth=0.25ex,yshift=0.5mm]
\tikzstyle{label}=[font=\footnotesize,text height=1.5ex, text depth=0.25ex,yshift=0.5mm]
\tikzstyle{left label}=[label,anchor=east,xshift=1.5mm]
\tikzstyle{right label}=[label,anchor=west,xshift=-1.5mm]

\tikzstyle{braceedge}=[decorate,decoration={brace,amplitude=2mm,raise=-1mm}]
\tikzstyle{small braceedge}=[decorate,decoration={brace,amplitude=1mm,raise=-1mm}]

\tikzstyle{doubled}=[line width=2pt] \tikzstyle{boldedge}=[doubled,shorten <=-0.17mm,shorten >=-0.17mm]

\tikzstyle{boldedgedashed}=[very thick,dashed,shorten <=-0.17mm,shorten >=-0.17mm]
\tikzstyle{vboldedgedashed}=[doubled,dashed,shorten <=-0.17mm,shorten >=-0.17mm]
\tikzstyle{left hook arrow}=[left hook-latex]
\tikzstyle{right hook arrow}=[right hook-latex]
\tikzstyle{sembracket}=[line width=0.5pt,shorten <=-0.07mm,shorten >=-0.07mm]

\tikzstyle{causal edge}=[->,thick,gray]
\tikzstyle{causal nondir}=[thick,gray]
\tikzstyle{timeline}=[thick,gray, dashed]

\tikzstyle{cedge}=[<->,thick,gray!70!white]

\tikzstyle{empty diagram}=[draw=gray!40!white,dashed,shape=rectangle,minimum width=1cm,minimum height=1cm]
\tikzstyle{empty diagram small}=[draw=gray!50!white,dashed,shape=rectangle,minimum width=0.6cm,minimum height=0.5cm]

\tikzstyle{dot}=[inner sep=0.7mm,minimum width=0pt,minimum height=0pt,draw,shape=circle]
\tikzstyle{ddot}=[inner sep=0.7mm,doubled, minimum width=2.5mm,minimum height=2.5mm,draw,shape=circle]

\tikzstyle{black dot}=[dot,fill=black]
\tikzstyle{white dot}=[dot,fill=white]
\tikzstyle{green dot}=[white dot] \tikzstyle{gray dot}=[dot,fill=gray!40!white]
\tikzstyle{red dot}=[gray dot]

\tikzstyle{black ddot}=[ddot,fill=black]
\tikzstyle{white ddot}=[ddot,fill=white]
\tikzstyle{gray ddot}=[ddot,fill=gray!40!white]

\tikzstyle{gray edge}=[gray!40!white]

\tikzstyle{small dot}=[inner sep=0.4mm,minimum width=0pt,minimum height=0pt,draw,shape=circle]

\tikzstyle{small black dot}=[small dot,fill=black]
\tikzstyle{small white dot}=[small dot,fill=white]
\tikzstyle{small gray dot}=[small dot,fill=gray!40!white]

\tikzstyle{causal dot}=[inner sep=0.4mm,minimum width=0pt,minimum height=0pt,draw=white,shape=circle,fill=gray!40!white]

\tikzstyle{white phase dot}=[dot,fill=white]
\tikzstyle{white phase ddot}=[ddot,fill=white]

\tikzstyle{gray phase dot}=[dot,fill=gray!40!white]
\tikzstyle{gray phase ddot}=[ddot,fill=gray!40!white]
\tikzstyle{grey phase dot}=[gray phase dot]
\tikzstyle{grey phase ddot}=[gray phase ddot]

\tikzstyle{cnot}=[fill=white,shape=circle,inner sep=-1.4pt]
\tikzstyle{hadamard}=[square box,inner sep=0 pt,font=\tiny\sf,minimum height=3mm,minimum width=3mm]
\tikzstyle{dhadamard}=[hadamard,doubled]
\tikzstyle{antipode}=[white dot,inner sep=0.3mm,font=\footnotesize]

\tikzstyle{scalar}=[diamond,draw,inner sep=0.5pt,font=\small]
\tikzstyle{dscalar}=[diamond,doubled, draw,inner sep=0.5pt,font=\small]

\tikzstyle{small box}=[rectangle,inline text,fill=white,draw,minimum height=5mm,yshift=-0.5mm,minimum width=5mm,font=\small]
\tikzstyle{small gray box}=[small box,fill=gray!30]
\tikzstyle{medium box}=[rectangle,inline text,fill=white,draw,minimum height=5mm,yshift=-0.5mm,minimum width=10mm,font=\small]
\tikzstyle{square box}=[small box] \tikzstyle{medium gray box}=[small box,fill=gray!30]
\tikzstyle{large box}=[rectangle,inline text,fill=white,draw,minimum height=5mm,yshift=-0.5mm,minimum width=15mm,font=\small]
\tikzstyle{large gray box}=[small box,fill=gray!30]

\tikzstyle{point}=[regular polygon,regular polygon sides=3,draw,scale=0.75,inner sep=-0.5pt,minimum width=9mm,fill=white,regular polygon rotate=180]
\tikzstyle{copoint}=[regular polygon,regular polygon sides=3,draw,scale=0.75,inner sep=-0.5pt,minimum width=9mm,fill=white]
\tikzstyle{dpoint}=[point,doubled]
\tikzstyle{dcopoint}=[copoint,doubled]

\tikzstyle{tinypoint}=[regular polygon,regular polygon sides=3,draw,scale=0.55,inner sep=-0.15pt,minimum width=6mm,fill=white,regular polygon rotate=180]

\tikzstyle{white point}=[point]
\tikzstyle{green point}=[white point] \tikzstyle{white copoint}=[copoint]
\tikzstyle{gray point}=[point,fill=gray!40!white]
\tikzstyle{gray dpoint}=[gray point,doubled]
\tikzstyle{red point}=[gray point] \tikzstyle{gray copoint}=[copoint,fill=gray!40!white]
\tikzstyle{gray dcopoint}=[gray copoint,doubled]

\tikzstyle{tiny gray point}=[tinypoint,fill=gray!40!white]

\tikzstyle{diredge}=[->]
\tikzstyle{rdiredge}=[<-]
\tikzstyle{thickdiredge}=[->, very thick]
\tikzstyle{pointer edge}=[->,very thick,gray]
\tikzstyle{pointer edge part}=[very thick,gray]
\tikzstyle{dashed edge}=[dashed]
\tikzstyle{thick dashed edge}=[very thick,dashed]
\tikzstyle{thick gray dashed edge}=[thick dashed edge,gray!90]
\tikzstyle{thick map edge}=[very thick,|->]

\makeatletter
\newcommand{\boxshape}[3]{\pgfdeclareshape{#1}{
\inheritsavedanchors[from=rectangle] \inheritanchorborder[from=rectangle]
\inheritanchor[from=rectangle]{center}
\inheritanchor[from=rectangle]{north}
\inheritanchor[from=rectangle]{south}
\inheritanchor[from=rectangle]{west}
\inheritanchor[from=rectangle]{east}
\backgroundpath{\southwest \pgf@xa=\pgf@x \pgf@ya=\pgf@y
\northeast \pgf@xb=\pgf@x \pgf@yb=\pgf@y

\@tempdima=#2
\@tempdimb=#3

\pgfpathmoveto{\pgfpoint{\pgf@xa - 5pt + \@tempdima}{\pgf@ya}}
\pgfpathlineto{\pgfpoint{\pgf@xa - 5pt - \@tempdima}{\pgf@yb}}
\pgfpathlineto{\pgfpoint{\pgf@xb + 5pt + \@tempdimb}{\pgf@yb}}
\pgfpathlineto{\pgfpoint{\pgf@xb + 5pt - \@tempdimb}{\pgf@ya}}
\pgfpathlineto{\pgfpoint{\pgf@xa - 5pt + \@tempdima}{\pgf@ya}}
\pgfpathclose
}
}}

\boxshape{NEbox}{0pt}{5pt}
\boxshape{SEbox}{0pt}{-5pt}
\boxshape{NWbox}{5pt}{0pt}
\boxshape{SWbox}{-5pt}{0pt}
\boxshape{EBox}{-3pt}{3pt}
\boxshape{WBox}{3pt}{-3pt}
\makeatother

\tikzstyle{cloud}=[shape=cloud,draw,minimum width=1.5cm,minimum height=1.5cm]

\tikzstyle{map}=[draw,shape=NEbox,inner sep=2pt,minimum height=6mm,fill=white]
\tikzstyle{mapdag}=[draw,shape=SEbox,inner sep=2pt,minimum height=6mm,fill=white]
\tikzstyle{mapadj}=[draw,shape=SEbox,inner sep=2pt,minimum height=6mm,fill=white]
\tikzstyle{maptrans}=[draw,shape=SWbox,inner sep=2pt,minimum height=6mm,fill=white]
\tikzstyle{mapconj}=[draw,shape=NWbox,inner sep=2pt,minimum height=6mm,fill=white]

\tikzstyle{dbox}=[draw,doubled,shape=rectangle,inner sep=2pt,minimum height=6mm,minimum width=6mm,fill=white]
\tikzstyle{dmap}=[draw,doubled,shape=NEbox,inner sep=2pt,minimum height=6mm,fill=white]
\tikzstyle{dmapdag}=[draw,doubled,shape=SEbox,inner sep=2pt,minimum height=6mm,fill=white]
\tikzstyle{dmapadj}=[draw,doubled,shape=SEbox,inner sep=2pt,minimum height=6mm,fill=white]
\tikzstyle{dmaptrans}=[draw,doubled,shape=SWbox,inner sep=2pt,minimum height=6mm,fill=white]
\tikzstyle{dmapconj}=[draw,doubled,shape=NWbox,inner sep=2pt,minimum height=6mm,fill=white]

\tikzstyle{ddmap}=[draw,doubled,dashed,shape=NEbox,inner sep=2pt,minimum height=6mm,fill=white]
\tikzstyle{ddmapdag}=[draw,doubled,dashed,shape=SEbox,inner sep=2pt,minimum height=6mm,fill=white]
\tikzstyle{ddmapadj}=[draw,doubled,dashed,shape=SEbox,inner sep=2pt,minimum height=6mm,fill=white]
\tikzstyle{ddmaptrans}=[draw,doubled,dashed,shape=SWbox,inner sep=2pt,minimum height=6mm,fill=white]
\tikzstyle{ddmapconj}=[draw,doubled,dashed,shape=NWbox,inner sep=2pt,minimum height=6mm,fill=white]

\boxshape{sNEbox}{0pt}{3pt}
\boxshape{sSEbox}{0pt}{-3pt}
\boxshape{sNWbox}{3pt}{0pt}
\boxshape{sSWbox}{-3pt}{0pt}
\tikzstyle{smap}=[draw,shape=sNEbox,fill=white]
\tikzstyle{smapdag}=[draw,shape=sSEbox,fill=white]
\tikzstyle{smapadj}=[draw,shape=sSEbox,fill=white]
\tikzstyle{smaptrans}=[draw,shape=sSWbox,fill=white]
\tikzstyle{smapconj}=[draw,shape=sNWbox,fill=white]

\tikzstyle{dsmap}=[draw,dashed,shape=sNEbox,fill=white]
\tikzstyle{dsmapdag}=[draw,dashed,shape=sSEbox,fill=white]
\tikzstyle{dsmaptrans}=[draw,dashed,shape=sSWbox,fill=white]
\tikzstyle{dsmapconj}=[draw,dashed,shape=sNWbox,fill=white]

\boxshape{mNEbox}{0pt}{10pt}
\boxshape{mSEbox}{0pt}{-10pt}
\boxshape{mNWbox}{10pt}{0pt}
\boxshape{mSWbox}{-10pt}{0pt}
\tikzstyle{mmap}=[draw,shape=mNEbox]
\tikzstyle{mmapdag}=[draw,shape=mSEbox]
\tikzstyle{mmaptrans}=[draw,shape=mSWbox]
\tikzstyle{mmapconj}=[draw,shape=mNWbox]

\tikzstyle{mmapgray}=[draw,fill=gray!40!white,shape=mNEbox]
\tikzstyle{smapgray}=[draw,fill=gray!40!white,shape=sNEbox]

\makeatletter
\pgfdeclareshape{cornerpoint}{
\inheritsavedanchors[from=rectangle] \inheritanchorborder[from=rectangle]
\inheritanchor[from=rectangle]{center}
\inheritanchor[from=rectangle]{north}
\inheritanchor[from=rectangle]{south}
\inheritanchor[from=rectangle]{west}
\inheritanchor[from=rectangle]{east}
\backgroundpath{\southwest \pgf@xa=\pgf@x \pgf@ya=\pgf@y
\northeast \pgf@xb=\pgf@x \pgf@yb=\pgf@y

\pgfmathsetmacro{\pgf@shorten@left}{\pgfkeysvalueof{/tikz/shorten left}}
\pgfmathsetmacro{\pgf@shorten@right}{\pgfkeysvalueof{/tikz/shorten right}}

\pgfpathmoveto{\pgfpoint{0.5 * (\pgf@xa + \pgf@xb)}{\pgf@ya - 5pt}}
\pgfpathlineto{\pgfpoint{\pgf@xa - 8pt + \pgf@shorten@left}{\pgf@yb - 1.5 * \pgf@shorten@left}}
\pgfpathlineto{\pgfpoint{\pgf@xa - 8pt + \pgf@shorten@left}{\pgf@yb}}
\pgfpathlineto{\pgfpoint{\pgf@xb + 8pt - \pgf@shorten@right}{\pgf@yb}}
\pgfpathlineto{\pgfpoint{\pgf@xb + 8pt - \pgf@shorten@right}{\pgf@yb - 1.5 * \pgf@shorten@right}}
\pgfpathclose
}
}

\pgfdeclareshape{cornercopoint}{
\inheritsavedanchors[from=rectangle] \inheritanchorborder[from=rectangle]
\inheritanchor[from=rectangle]{center}
\inheritanchor[from=rectangle]{north}
\inheritanchor[from=rectangle]{south}
\inheritanchor[from=rectangle]{west}
\inheritanchor[from=rectangle]{east}
\backgroundpath{\southwest \pgf@xa=\pgf@x \pgf@ya=\pgf@y
\northeast \pgf@xb=\pgf@x \pgf@yb=\pgf@y

\pgfmathsetmacro{\pgf@shorten@left}{\pgfkeysvalueof{/tikz/shorten left}}
\pgfmathsetmacro{\pgf@shorten@right}{\pgfkeysvalueof{/tikz/shorten right}}

\pgfpathmoveto{\pgfpoint{0.5 * (\pgf@xa + \pgf@xb)}{\pgf@yb + 5pt}}
\pgfpathlineto{\pgfpoint{\pgf@xa - 8pt + \pgf@shorten@left}{\pgf@ya + 1.5 * \pgf@shorten@left}}
\pgfpathlineto{\pgfpoint{\pgf@xa - 8pt + \pgf@shorten@left}{\pgf@ya}}
\pgfpathlineto{\pgfpoint{\pgf@xb + 8pt - \pgf@shorten@right}{\pgf@ya}}
\pgfpathlineto{\pgfpoint{\pgf@xb + 8pt - \pgf@shorten@right}{\pgf@ya + 1.5 * \pgf@shorten@right}}
\pgfpathclose
}
}

\makeatother

\pgfkeyssetvalue{/tikz/shorten left}{0pt}
\pgfkeyssetvalue{/tikz/shorten right}{0pt}

\tikzstyle{kpoint common}=[draw,fill=white,inner sep=1pt,minimum height=4mm]
\tikzstyle{kpoint}=[shape=cornerpoint,shorten left=5pt,kpoint common]
\tikzstyle{kpoint adjoint}=[shape=cornercopoint,shorten left=5pt,kpoint common]
\tikzstyle{kpoint conjugate}=[shape=cornerpoint,shorten right=5pt,kpoint common]
\tikzstyle{kpoint transpose}=[shape=cornercopoint,shorten right=5pt,kpoint common]
\tikzstyle{kpoint symm}=[shape=cornerpoint,shorten left=5pt,shorten right=5pt,kpoint common]

\tikzstyle{kpointdag}=[kpoint adjoint]
\tikzstyle{kpointadj}=[kpoint adjoint]
\tikzstyle{kpointconj}=[kpoint conjugate]
\tikzstyle{kpointtrans}=[kpoint transpose]

\tikzstyle{dkpoint}=[kpoint,doubled]
\tikzstyle{dkpointdag}=[kpoint adjoint,doubled]
\tikzstyle{dkcopoint}=[kpoint adjoint,doubled]
\tikzstyle{dkpointadj}=[kpoint adjoint,doubled]
\tikzstyle{dkpointconj}=[kpoint conjugate,doubled]
\tikzstyle{dkpointtrans}=[kpoint transpose,doubled]

\tikzstyle{kscalar}=[kpoint common, shape=EBox, inner xsep=-1pt, inner ysep=3pt,font=\small]
\tikzstyle{kscalarconj}=[kpoint common, shape=WBox, inner xsep=-1pt, inner ysep=3pt,font=\small]

 \tikzstyle{upground}=[circuit ee IEC,thick,ground,rotate=90,scale=2.5]
 \tikzstyle{downground}=[circuit ee IEC,thick,ground,rotate=-90,scale=2.5]
  \tikzstyle{bigground}=[regular polygon,regular polygon sides=3,draw=gray,scale=0.50,inner sep=-0.5pt,minimum width=10mm,fill=gray]

\tikzstyle{arrs}=[-latex,font=\small,auto]
\tikzstyle{arrow plain}=[arrs]
\tikzstyle{arrow dashed}=[dashed,arrs]
\tikzstyle{arrow bold}=[very thick,arrs]
\tikzstyle{arrow hide}=[draw=white!0,-]
\tikzstyle{arrow reverse}=[latex-]
\tikzstyle{cdnode}=[]

 \tikzstyle{slit}=[line width=2]
\tikzstyle{block}=[line width=4,gray,line cap=round]
\tikzstyle{screen}=[line width=4,black,line cap=round]
\tikzstyle{di}=[diamond,draw,inner sep=0.5pt,font=\small, minimum size = .5cm]
\tikzstyle{sbox}=[rectangle,draw]
\tikzstyle{mirror}=[line width=2,black]
\tikzstyle{trace}=[circuit ee IEC,thick,ground,rotate=0,scale=2]
\tikzstyle{traceState}=[circuit ee IEC,thick,ground,rotate=180,scale=2]
\tikzstyle{detEff}=[circuit ee IEC,thick,ground,rotate=180,scale=1.4]
\tikzstyle{maxMix}=[circuit ee IEC,thick,ground,scale=1.4]
\tikzstyle{particlePath}=[line width=2,gray!40, line cap =round]
\usepackage{fp}
\usetikzlibrary{calc,decorations.pathmorphing,decorations.text,fixedpointarithmetic}

\newcommand{\diageq}{\rotatebox{-45}{$\,=$}}
\newcommand{\detEff}{(\begin{tikzpicture}
	\begin{pgfonlayer}{nodelayer}
		\node [style=detEff] (0) at (0, -0) {};
		\node [style=none] (1) at (0.15, -0) {};
	\end{pgfonlayer}	
\end{tikzpicture}|}
\newcommand{\maxMix}{|\begin{tikzpicture}
	\begin{pgfonlayer}{nodelayer}
		\node [style=maxMix] (0) at (0, -0) {};
		\node [style=none] (1) at (0.1, -0) {};
	\end{pgfonlayer}
\end{tikzpicture})}

\begin{document}

\title{Generalised phase kick-back: the structure of computational algorithms from physical principles}
\author{Ciar{\'a}n~M. Lee}
\email{ciaran.lee@cs.ox.ac.uk}
\affiliation{University of Oxford, Department of Computer Science, Wolfson Building, Parks Road, Oxford OX1 3QD, UK.}
\author{John~H. Selby}
\email{john.selby08@imperial.ac.uk}
\affiliation{University of Oxford, Department of Computer Science, Wolfson Building, Parks Road, Oxford OX1 3QD, UK.}
\affiliation{Imperial College London,  London SW7 2AZ, UK.}

\begin{abstract}
The advent of quantum computing has challenged classical conceptions of which problems are efficiently solvable in our physical world. This motivates the general study of how physical principles bound computational power.
In this paper we show that some of the essential machinery of quantum computation -- namely reversible controlled transformations and the phase kick-back mechanism -- exist in any operational-defined theory with a consistent notion of information. These results provide the tools for an exploration of the physics underpinning the structure of computational algorithms.
We use these results to investigate the relationship between interference behaviour and computational power, demonstrating that non-trivial interference behaviour is a general resource for post-classical computation. In proving the above, we connect post-quantum interference -- the higher-order interference of Sorkin -- to the existence of post-quantum particle types, potentially providing a novel experimental test for higher-order interference. Finally, we conjecture that theories with post-quantum interference can solve problems intractable even on a quantum computer.
\end{abstract}

\maketitle

\section{Introduction}

One of the major conceptual breakthroughs in physics over the past thirty years was the realisation that quantum theory offers dramatic advantages \cite{Nielsen} for various information-processing tasks -- computation in particular \cite{shor, arkhipov, Nielsen}. This raises the general question of how physical principles bound computational power. Moreover, what broad relationships exist between such principles and computation? A major roadblock to such an investigation is that quantum computation is phrased in the language of Hilbert spaces, which lacks direct physical or operational significance.

In contrast, the framework of operationally-defined theories \cite{Pavia1, Pavia2, Hardy-2011, Barrett-2007,LB-2014} provides a clear-cut operational language in which to investigate this problem. Theories within this framework can differ \cite{Barrett-2007} from classical and quantum theories. Whilst many of them may not correspond to descriptions of our physical world, they make good operational sense and allow one to assess how computational power depends on the physical principles underlying them in a systematic manner.

Previous investigations into computation within this framework have taken a high-level approach using the language of complexity classes to derive general bounds on the power of computation \cite{LB-2014, Proofs, landscape}. However, much of quantum computing is concerned not so much with this high-level view, but instead with the construction of concrete algorithms to solve specific problems. A deeper understanding of the general structure of computational algorithms in this framework has so far remained illusive. Here we take this low-level algorithmic view and ask which physical principles are required to allow for some of the common machinery of quantum computation in this context.

In this paper we show that three physical principles, \emph{causality} (which roughly states that information propagates from present to future), \emph{purification} (roughly, that information is fundamentally conserved) and \emph{strong symmetry} (all information carriers of the same size are equivalent)
-- which are necessary for a well defined notion of information -- are sufficient for the existence of reversible controlled transformations (Thm.~(\ref{Reversible-Control}), Sec.~(\ref{Control2})) and a generalised \emph{phase kick-back mechanism} (Thm.~(\ref{ALL}), Sec.~(\ref{Control2})). In the quantum case, the phase kick-back mechanism \cite{kickback} plays a vital role in almost all algorithms -- notably the Deutsch-Jozsa algorithm, Grover's search algorithm and Simon's algorithm -- whilst reversible controlled transformations are central components of most well-studied universal gate sets and fundamental for the definition of computational oracles.

One might ask how the computational power of theories with these crucial algorithmic components depends on their underlying physical properties. One such property -- currently under both theoretical \cite{Higher-order-reconstruction,Niestegge-2012,Henson-2015,ududec2011three} and experimental \cite{sinha2008testing,park2012three} investigation -- is the existence of \emph{higher-order interference}.

Sorkin \cite{sorkin1994quantum, sorkin1995quantum} has introduced a hierarchy of mathematically conceivable \emph{higher-order} interference behaviours and shown that quantum theory is limited to having only second-order interference. Informally, this means that the interference pattern created in a three -- or more -- slit experiment can be written in terms of the two and one slit interference patterns obtained by blocking some of the slits; no genuinely new features result from considering three slits instead of two. This is in contrast to the existence of second-order interference where the two slit interference cannot be reproduced from that of single slits. Informally, theories are said to have higher-order interference if irreducible interference patterns can be created in multi-slit experiments.

Second-order interference between quantum computational paths appears to be a resource for non-classical computation \cite{Interference-speed-up, Nielsen}. It therefore seems prudent to investigate how different interference behaviour is related to computation in general. In quantum theory there is an intimate connection between phase transformations -- such as those used in the kick-back mechanism -- and interference. Motivated by this, in Sec.~(\ref{int}), we introduce a framework that relates higher-order interference to \emph{phase transformations} in operationally-defined theories.

We show that the generalised phase kick-back mechanism allows one to access any `higher-order phase' in a controlled manner. Using this, in Sec.~(\ref{oracle}), we show that the existence of non-trivial interference behaviour allows for the solution of problems intractable on a classical computer. We also conjecture that these higher-order phase kick-backs allow for the solution of computational problems intractable even on a \emph{quantum} computer. Additionally, in Sec.~(\ref{Exchange}), we show that higher-order phases lead to new particle types that exhibit both qualitatively and quantitatively different behaviour to fermions, bosons and anyons. Thus potentially providing a new experimental test of higher-order interference.

\section{The framework}
\subsection{Operational physical theories}

We work in the circuit framework for operationally-defined theories developed in \cite{Hardy-2011,Pavia1,Pavia2}.
An operational theory specifies a set of physical processes that can be connected together to form experiments and assigns probabilities to different experimental outcomes. A process has input ports, output ports, and a classical pointer. When a process is used in an experiment, the pointer comes to rest in one of a number of positions, indicating an outcome has occurred. Intuitively, one can think of \emph{physical systems} as passing between the ports of these processes. These systems come in different types, denoted $A,B..$. In an experiment these processes can be composed both sequentially and in parallel, and when composed sequentially, types must match.

In this framework, closed circuits define probabilities. Processes that yield the same probabilities in all closed circuits are identified. The set of equivalence classes of processes with no input ports are called \emph{states}, no output ports \emph{effects} and both input and output ports \emph{transformations}. The set of all states of system $A$ is denoted $\Omega_A$, the set of all effects on $B$ is denoted $\mathcal{E}_B$ and the set of \emph{reversible} transformations between systems $A$ and $B$ is denoted $\mathcal{R}^A_B$ \footnote{The set of states, effects and transformations each give rise to a vector space and transformations and effects act linearly on the vector space of states. We assume in this work that all vector spaces are finite dimensional.}. Note that $\mathcal{R}^A_B$ has a group structure. A state is \emph{pure} if it does not arise as a \emph{coarse-graining} of other states \footnote{The process $\{\mathcal{U}_j\}_{j\in{Y}}$, where $j$ index the positions of the classical pointer, is a coarse-graining of the process $\{\mathcal{E}_i\}_{i\in{X}}$ if there is a disjoint partition $\{X_j\}_{j\in{Y}}$ of $X$ such that $\mathcal{U}_j=\sum_{i\in{X_j}}\mathcal{E}_i$.}; a pure state is one for which we have maximal information. A state is \emph{mixed} if it is not pure. We assume for this paper that the composite of two pure states is itself pure \footnote{Note that this is not the case for every generalised probabilistic theories. For example, the theories based on Euclidean-Jordan algebras presented in \cite{barnum2015some} do not satisfy this requirement.}. Similarly, one says a transformation is pure if it does not arise as a coarse-graining of other transformations. It can be shown that reversible transformations preserve pure states.
 
The `Dirac-like' notation $_A|s)$ is used to represent a state of system $A$, and $(e_{r}|_B$ to represent an effect on $B$. Here $r$ is the position of the classical pointer, which can be thought of as the outcome of the measurement defined by $\{(e_r|\}_r$. States, effects and transformations can be represented diagrammatically:
\[\begin{tikzpicture}
	\begin{pgfonlayer}{nodelayer}
		\node [style=cpoint] (0) at (-0.5, -0) {$s$};
		\node [style={small box}] (1) at (1.5, -0) {$T$};
		\node [style=cocpoint] (2) at (3.5, -0) {$e_r$};
		\node [style=none] (3) at (4.5, -0) {$=$};
		\node [style=none] (4) at (6.5, -0) {$(e_r|_BT_A|s)$};
		\node [style=none] (5) at (0.5, 0.5) {$A$};
		\node [style=none] (6) at (2.5, 0.5) {$B$};
		\node [style=none] (7) at (8.5, -0) {};
	\end{pgfonlayer}
	\begin{pgfonlayer}{edgelayer}
		\draw (0) to (1);
		\draw (1) to (2);
	\end{pgfonlayer}
\end{tikzpicture} \]
This diagrammatic approach was inspired by the categorical formalism of quantum mechanics \cite{cqm1,cqm2}.
\begin{definition}[Causality \cite{Pavia1}]
A theory is said to be \emph{causal} if there exists a unique deterministic effect $\detEff $  for every system, such that $\sum_r (e_r|=\detEff $ for all measurements, $\{(e_r|\}_r$.
\end{definition}

Mathematically, causality is equivalent to the statement: ``Probabilities of present experiments are independent of future measurement choices''. In causal theories, all states are \emph{normalised} \cite{Pavia1}. That is, $\detEff s)=1$ for all $|s)$. The deterministic effect allows one to define a notion of \emph{marginalisation} for multi-partite states.

\begin{definition}[Purification \cite{Pavia1}] \label{Pure}
Given a state $_A|s)$ there exists a system $B$ and a pure state $_{AB}|\psi)$ on $AB$ such that $_A|s)$ is the marginalisation of $_{AB}|\psi)$: \[\begin{tikzpicture}
	\begin{pgfonlayer}{nodelayer}
		\node [style=none] (0) at (1, -0) {};
		\node [style=none] (1) at (1, 1.5) {};
		\node [style=none] (2) at (1, 2) {};
		\node [style=none] (3) at (1, -0.5) {};
		\node [style=trace] (4) at (2.5, -0) {};
		\node [style=none] (5) at (2.5, 1.5) {};
		\node [style=none] (6) at (1, -0.5) {};
		\node [style=none] (7) at (0.5, 0.75) {$\psi$};
		\node [style=none] (8) at (4, 1) {$=$};
		\node [style=cpoint] (9) at (5, 1.5) {$s$};
		\node [style=none] (10) at (6.5, 1.5) {};
		\node [style=none] (11) at (1.75, 0.5) {$B$};
		\node [style=none] (12) at (1.75, 2) {$A$};
		\node [style=none] (13) at (5.75, 2) {$A$};
		\node [style=none] (14) at (7, -0) {};
	\end{pgfonlayer}
	\begin{pgfonlayer}{edgelayer}
		\draw [bend right=90, looseness=1.25] (2.center) to (3.center);
		\draw (2.center) to (3.center);
		\draw (1.center) to (5.center);
		\draw (0.center) to (4);
		\draw (9) to (10.center);
	\end{pgfonlayer}
\end{tikzpicture}  \] Moreover, the purification $_{AB}|\psi)$ is unique up to reversible transformations on the purifying system, $B$ \footnote{That is if two states $|\psi)_{AB}$ and $|\psi')_{AB}$ purify $|s)_A$, then there exists a reversible transformation $T_B$ on system $B$ such that $|\psi)_{AB}=(\mathbb{I}\otimes{T_B})|\psi)_{AB}$.}.
\end{definition}
While the purification principle appears to only concern states, it can be leveraged to prove somewhat analogous results about transformations \cite[Thm.~15]{Pavia1}: let $T,T'$ be reversible transformations. If,
\[\begin{tikzpicture}
	\begin{pgfonlayer}{nodelayer}
		\node [style=none] (0) at (0, 0.75) {};
		\node [style=none] (1) at (1, 0.75) {};
		\node [style=none] (2) at (1, 1.25) {};
		\node [style=none] (3) at (2, 1.25) {};
		\node [style=none] (4) at (2, 0.75) {};
		\node [style=none] (5) at (3, 0.75) {};
		\node [style=none] (6) at (1, -0.25) {};
		\node [style=cpoint] (7) at (0, -0.25) {s};
		\node [style=none] (8) at (1, -0.75) {};
		\node [style=none] (9) at (2, -0.75) {};
		\node [style=none] (10) at (2, -0.25) {};
		\node [style=trace] (11) at (3, -0.25) {};
		\node [style=none] (12) at (1.5, 0.25) {$T$};
		\node [style=none] (13) at (4.5, 0.25) {$=$};
		\node [style=none] (14) at (7, -0.75) {};
		\node [style=none] (15) at (6, 0.75) {};
		\node [style=none] (16) at (8, -0.25) {};
		\node [style=none] (17) at (7, 0.75) {};
		\node [style=none] (18) at (9, 0.75) {};
		\node [style=none] (19) at (8, 1.25) {};
		\node [style=none] (20) at (7, -0.25) {};
		\node [style=none] (21) at (7, 1.25) {};
		\node [style=trace] (22) at (9, -0.25) {};
		\node [style=none] (23) at (7.5, 0.25) {$T'$};
		\node [style=none] (24) at (8, -0.75) {};
		\node [style=none] (25) at (8, 0.75) {};
		\node [style=cpoint] (26) at (6, -0.25) {s'};
		\node [style=none] (27) at (10, -0) {};
	\end{pgfonlayer}
	\begin{pgfonlayer}{edgelayer}
		\draw (0.center) to (1.center);
		\draw (7) to (6.center);
		\draw (10.center) to (11);
		\draw (4.center) to (5.center);
		\draw (3.center) to (9.center);
		\draw (9.center) to (8.center);
		\draw (8.center) to (2.center);
		\draw (2.center) to (3.center);
		\draw (15.center) to (17.center);
		\draw (26) to (20.center);
		\draw (16.center) to (22);
		\draw (25.center) to (18.center);
		\draw (19.center) to (24.center);
		\draw (24.center) to (14.center);
		\draw (14.center) to (21.center);
		\draw (21.center) to (19.center);
	\end{pgfonlayer}
\end{tikzpicture} ,\]
then there exists a reversible transformation $G$ such that
\begin{equation}\begin{tikzpicture}
	\begin{pgfonlayer}{nodelayer}
		\node [style=cpoint] (0) at (6, 1) {$\sigma$};
		\node [style=none] (1) at (8, 1.5) {};
		\node [style=none] (2) at (4.5, 0.25) {$=$};
		\node [style=none] (3) at (1, 1) {};
		\node [style=none] (4) at (2, 1.5) {};
		\node [style=cpoint] (5) at (0, -0.5) {s};
		\node [style=none] (6) at (1.5, 0.25) {$T$};
		\node [style=none] (7) at (8, -0.5) {};
		\node [style=cpoint] (8) at (6, -0.5) {s'};
		\node [style=none] (9) at (1, 1.5) {};
		\node [style=none] (10) at (2, -0.5) {};
		\node [style=none] (11) at (2, 1) {};
		\node [style=none] (12) at (7.5, 0.25) {$T'$};
		\node [style=none] (13) at (1, -1) {};
		\node [style=none] (14) at (3, 1) {};
		\node [style=none] (15) at (7, 1.5) {};
		\node [style=cpoint] (16) at (0, 1) {$\sigma$};
		\node [style=none] (17) at (1, -0.5) {};
		\node [style=none] (18) at (2, -1) {};
		\node [style=none] (19) at (10, 1) {};
		\node [style=none] (20) at (3, -0.5) {};
		\node [style=none] (21) at (7, -1) {};
		\node [style=none] (22) at (7, 1) {};
		\node [style=none] (23) at (8, 1) {};
		\node [style={small box}] (24) at (9, -0.5) {G};
		\node [style=none] (25) at (8, -1) {};
		\node [style=none] (26) at (10, -0.5) {};
		\node [style=none] (27) at (7, -0.5) {};
		\node [style=none] (28) at (12, -0) {$\forall |\sigma)$};
		\node [style=none] (29) at (13, -0) {};
	\end{pgfonlayer}
	\begin{pgfonlayer}{edgelayer}
		\draw (16) to (3.center);
		\draw (5) to (17.center);
		\draw (10.center) to (20.center);
		\draw (11.center) to (14.center);
		\draw (4.center) to (18.center);
		\draw (18.center) to (13.center);
		\draw (13.center) to (9.center);
		\draw (9.center) to (4.center);
		\draw (0) to (22.center);
		\draw (8) to (27.center);
		\draw (7.center) to (24);
		\draw (23.center) to (19.center);
		\draw (1.center) to (25.center);
		\draw (25.center) to (21.center);
		\draw (21.center) to (15.center);
		\draw (15.center) to (1.center);
		\draw (24) to (26.center);
	\end{pgfonlayer}
\end{tikzpicture} .\label{Dilation}\end{equation}
Eq.~(\ref{Dilation}), above, depicts the equivalence of purifications of a transformation up to a local reversible transformation, as opposed to the purification of states mentioned in Def.~(\ref{Pure}). These can in fact been shown to be equivalent \cite{Pavia1, Pavia2}, and so we will use the term \emph{purification} when referring to either notion.

Pure states $\{|s_i)\}_{i=1}^n$ are \emph{perfectly distinguishable} if there exists a measurement, corresponding to effects $\{(e_j|\}_{j=1}^n$, such that $(e_j|s_i)=\delta_{ij}$ for all $i,j$. Note that an $n$-tuple of pure and perfectly distinguishable states can reliably encode an $n$-level classical system.
\begin{definition}[Strong symmetry \cite{Higher-order-reconstruction}]
A theory satisfies \emph{strong symmetry} if for any two $n$-tuples of pure and perfectly distinguishable states $\{|\rho_i)\},\{|\sigma_i)\},$ there exists a reversible transformation $T$ such that $T|\rho_i)=|\sigma_i)$ for $i=1,\dots,n$.
\end{definition}

Informally, the purification principle says that information is fundamentally conserved, strong symmetry states that all information carriers of the same size are equivalent and causality implies that information propagates from present to future. Note that standard quantum theory, real vector space quantum theory and the classical theory of pure states satisfy all of the above principles. These principles will be shown to be a primer for interesting and consistent computation. In Sec.~\ref{oracle}, we shall investigate the change in computational power as one varies the interference behaviour in theories satisfying causality, purification and strong symmetry.

\subsection{Higher-order interference via phase transformations}
\subsubsection{A quantum example} \label{Quantum-Example}
Perhaps the cleanest example of interference in quantum theory is exhibited by the Mach-Zehnder interferometer, illustrated below:
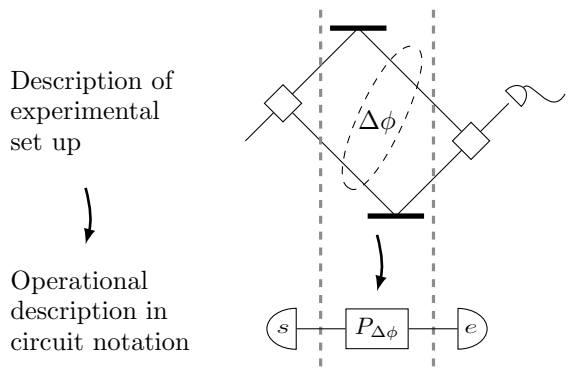
\[\begin{tikzpicture}
	\begin{pgfonlayer}{nodelayer}
		\node [style=none] (0) at (0, 1.25) {};
		\node [style=di, fill=white] (1) at (0.9999999, 2.25) {};
		\node [style=none] (2) at (3, 4.25) {};
		\node [style=none] (3) at (4, -0.7500001) {};
		\node [style=none] (4) at (7, 2.25) {};
		\node [style=di, fill=white] (5) at (6, 1.25) { };
		\node [style=none] (6) at (2, 4.75) {};
		\node [style=none] (7) at (2, -4.75) {};
		\node [style=none] (8) at (5, 4.75) {};
		\node [style=none] (9) at (5, -4.75) {};
		\node [style=cpoint] (10) at (0.9999999, -3.75) {$s$};
		\node [style=sbox, scale=1] (11) at (3.5, -3.75) {\small $P_{\Delta\phi}$};
		\node [style=cocpoint] (12) at (6, -3.75) {$e$};
		\node [style=none] (13) at (3.5, -1.25) {};
		\node [style=none] (14) at (3.5, -2.75) {};
		\node [style=none] (15) at (2.25, 4.25) {};
		\node [style=none] (16) at (3.75, 4.25) {};
		\node [style=none] (17) at (3.25, -0.7500001) {};
		\node [style=none] (18) at (4.75, -0.7500001) {};
		\node [style=mopoint, rotate=135, xscale=0.75, yscale=1] (19) at (7.25, 2.5) {};
		\node [style=none] (20) at (8.5, 2.5) {};
		\node [style=none] (21) at (4.5, 3.75) {};
		\node [style=none] (22) at (2.75, -0) {};
		\node [style=none] (23) at (3.5, 1.75) {$\Delta\phi$};
		\node [style=none, text width={2.5 cm}] (24) at (-3.75, 2) {Description of experimental set up};
		\node [style=none, text width={2.5 cm}] (25) at (-3.75, -3.25) {Operational description in circuit notation};
		\node [style=none] (26) at (-4.25, -0) {};
		\node [style=none] (27) at (-4.25, -1.5) {};
	\end{pgfonlayer}
	\begin{pgfonlayer}{edgelayer}
		\draw (0.center) to (2.center);
		\draw (1) to (3.center);
		\draw (3.center) to (4.center);
		\draw [style={thick gray dashed edge}] (7.center) to (6.center);
		\draw [style={thick gray dashed edge}] (9.center) to (8.center);
		\draw (10) to (11);
		\draw (11) to (12);
		\draw [style={arrow plain}, line width=1.00, bend left=15, looseness=1.00] (13.center) to (14.center);
		\draw [style=none, in=-120, out=30, looseness=1.25] (19) to (20.center);
		\draw [style=mirror] (15.center) to (16.center);
		\draw [style=mirror] (17.center) to (18.center);
		\draw (2.center) to (5);
		\draw [style={dashed edge}, bend left=90, looseness=0.50] (22.center) to (21.center);
		\draw [style={dashed edge}, bend right=105, looseness=0.50] (22.center) to (21.center);
		\draw [style={arrow plain}, line width=1.00, bend left=15, looseness=1.00] (26.center) to (27.center);
	\end{pgfonlayer}
\end{tikzpicture} \]
There are three parts to this:
\begin{enumerate}
\item Prepare a state as a superposition of paths: $$|s)=\ket{+}\bra{+}:=\rho_+$$
\item Apply a `phase transformation': $$P_{\Delta\phi}|s)=R_z^{\Delta\phi}\rho_+R_z^{\Delta\phi\dagger},$$ with $R_z^{\Delta\phi}$ a rotation by $\Delta\phi$ about the $z$ axis of the Bloch ball.
\item Measure in a superposition of paths: $$ \begin{aligned}  (e|P_{\Delta\phi}|s)&=\mathrm{Tr}\left(\rho_+R_z^{\Delta\phi}\rho_+R_z^{\Delta\phi\dagger}\right)=\cos^2\left(\frac{\Delta\phi}{2}\right). \end{aligned}$$
\end{enumerate}

The observed interference pattern is therefore a map from the group of `phase transformations', parametrised by $\Delta\phi$, to the unit interval (i.e. probabilities),
\begin{equation} \label{Quantum-Pattern}
P_{\Delta\phi}\mapsto \cos^2\left(\frac{\Delta\phi}{2}\right).
\end{equation}

The existence of interference in quantum theory is encapsulated in the statement: ``the interference pattern observed for a particular superposition measurement cannot be reproduced by the statistics generated by `which path' measurements''. In the above example this translates to:
\begin{equation} \label{Quantum-Interference}
\cos^2\left(\frac{\Delta\phi}{2}\right)\neq \sum_{i=0}^1 q_i \mathrm{Tr}\left(\ket{i}\bra{i}R_z^{\Delta\phi}\rho_+R_z^{\Delta\phi\dagger}\right),
\end{equation}
where $q_i$ is an arbitrary constant. Eq.~(\ref{Quantum-Interference}) is to be interpreted as an inequality of the functions defined on the right and left hand side. That is, these functions do not coincide on all phase transformations. This follows from the fact that:
\begin{equation} \label{Phase}
R_z^{\Delta\phi\dagger}\ket{i}\bra{i}R_z^{\Delta\phi}=\ket{i}\bra{i}, \quad \forall i\in\{0,1\}.
\end{equation}
That is, the left hand side of Eq.~(\ref{Quantum-Interference}) depends on $\Delta\phi$ whilst the right hand side does not.

\subsubsection{Operational theories \label{int} }

The quantum example from Sec.~(\ref{Quantum-Example}) illustrates the key components necessary to discuss interference:
\begin{itemize}
\item[(i)] a notion of `path',
\item[(ii)] a notion of `superposition of paths',
\item[(iii)] transformations that leave the statistics of `which path' measurements invariant, i.e. `phase transformations',
\item[(iv)] a notion of `interference pattern', i.e. a way of associating phase transformations with probabilities.
 \end{itemize}
  These points will now be discussed in the context of arbitrary operationally-defined theories. We then use this framework to link higher-order interference and phase transformations. Our approach is similar in spirit to that of Garner et al. \cite{Garner}, with the caveat that they have not considered higher-order interference.

\paragraph*{---(i) Paths:} A path is defined by a state and effect pair, where we view the state as `preparing a state which belongs to the path' and the effect as `measuring whether the state belongs to the path' and so we demand the probability of the state-effect pair to be one.
\begin{definition}{Paths,   $p$:} \[p:=(|s),(e|) \text{ s.t. } (e|s)=1.\]
\end{definition}

In our quantum example, the paths were $p_0=\left(\ket{0}\bra{0}, \ket{0}\bra{0}\right)$ and $p_1=\left(\ket{1}\bra{1}, \ket{1}\bra{1}\right)$.

Paths are disjoint if the state defining one path has zero probability of belonging to the other, and vice versa.
\begin{definition}{Disjoint paths, $p_1\perp p_2$:} \[p_1\perp p_2 \iff (e_i|s_j)=\delta_{ij}.\]
\end{definition}
An $n$-path experiment is defined by $n$ mutually disjoint paths such that the set consisting of the effects from each path forms a measurement.
\begin{definition}{$n$-path experiment,  $\mathds{P}$:} \[\mathds{P}:= \{p_i\} \text{ s.t. } p_i\perp p_j \ \forall i\neq j, \ \mathrm{and} \ \sum_i (e_i|=\detEff.\]
\end{definition}
In the quantum case, an $n$-path experiment would correspond to a multi-arm interferometer.

\paragraph*{---(ii) Superposition of paths:}
A superposition of paths will be defined relative to some $n$-path experiment $\mathds{P}$ via the notion of \emph{support}. We say that a state (or effect) has support on a path if the effect (or state) associated to that path gives a non-zero probability.
\begin{definition}{Support of a state or effect, $Supp[|s)]$ or $Supp[(e|]:$}
$$\begin{aligned}
Supp[|s)]:=\{p_i\in \mathds{P} \ | \ (e_i|s)\neq 0\},\\
Supp[(e|]:=\{p_i\in \mathds{P} \ | \ (e|s_i)\neq 0\}.
\end{aligned}
$$
\end{definition}
If the support of a state consists of more than one path this does not guarantee that it is a superposition of paths -- it could equally well be a classical mixture of paths. A superposition state must therefore lie outside the convex hull of the states which have support only on a single path. In our quantum example, the state $\ket{+}\bra{+}$ -- introduced in point~(1) of Sec.~(\ref{Quantum-Example}) -- was a superposition of paths.

We can define set of states (or effects) with support on some subset of paths $I\subseteq \mathds{P}$ as:
$$ \begin{aligned}
\Omega_I&:=\{|s)\in \Omega \ | \ Supp[|s)]=I\}, \\
\mathcal{E}_I&:=\{(e|\in \mathcal{E} \ | \ Supp[(e|]=I\}.
\end{aligned}
$$

\paragraph*{---(iii) Phase transformations:}

A phase transformation -- relative to some $\mathds{P}$ -- is any transformation that leaves the statistics of `which path' measurements invariant.
\begin{definition}{Phase group, $\mathcal{P}$:}
\[\mathcal{P}:=\{T\in \mathcal{R} \ | \ (e_i|T=(e_i|, \  \forall i\in \mathds{P}\}\]
\end{definition}
In the quantum example, the phase transformation was the rotation $R_z^{\Delta\phi}$ introduced in point~(2) of Sec.~(\ref{Quantum-Example}).

\paragraph*{---(iv) Interference patterns:}
We now generalise the quantum interference pattern of Eq.~(\ref{Quantum-Pattern}) to arbitrary operational theories.
\begin{definition}{Interference pattern, $\mathcal{C}_{s,e}:$}
\[\mathcal{C}_{s,e}: \ \mathcal{P}\to[0,1] \ :: \ T\mapsto (e|T|s)\]
\end{definition}

Given this definition, Eq.~(\ref{Quantum-Interference}) translates into the existence of $(e|\in\mathcal{E}_{\{0,1\}}$ -- that is, an effect with support on path $0$ and path $1$ -- and $|s)\in\Omega_{\{0,1\}}$ such that
\begin{equation} \label{General-Interference}
{C}_{s,e}\neq \sum_{i=0}^1 \mathcal{C}_{s,e_i}
\end{equation}
for all possible choices of $(e_i|\in\mathcal{E}_{\{i\}}$ including subnormalised effects, this is the analogue of the $q_i$'s in Eq.~(\ref{Quantum-Interference}). In other words, there is some choice of superposition state and effect such that their interference pattern cannot be reproduced by the statistics generated by effects with support on a single path.

Other approaches to defining higher-order interference in operational theories (for example \cite{Higher-order-reconstruction}) have additional structure such that one can define a set of `filters', $\{F_I\}$, for the theory. These are transformations that represent the action of leaving open some subset of paths $I$ whilst blocking the others. In this case one can define $(e_I|=(e|F_{I}$ giving a specific set of effects. However, arbitrary theories do not have sufficient structure to define filters and so one must consider all possible choices $(e_I|$ with the correct support. Otherwise \cite{Thesis, LS-2015} one can -- even in quantum and classical theory --  choose a specific set of $(e_I|$ to give the artificial appearance of higher-order interference.

It follows that the existence of a non-trivial phase group implies the existence of interference in a general theory. Indeed, the left hand side of Eq.~(\ref{General-Interference}) depends on the phase group element, whilst the right hand side does not -- the analogue of Eq.~(\ref{Phase}) from the quantum example. We now use our framework to discuss \emph{higher-order} interference.

\subsubsection{Higher-order interference and phase}

Adapting Sorkin's original definition of higher-order interference \cite{sorkin1994quantum} to our framework results in: the existence of $n$th-order interference in an $n$-path experiment corresponds to the existence of an effect $|e)$ and a state $(s|$ such that
\begin{equation} \label{Higher-order}
\mathcal{C}_{s,e}\neq \sum_{I\subset\mathds{P}}(-1)^{n-|I|+1}\mathcal{C}_{s,e_I},
\end{equation}
for all $|e_I)\in\mathcal{E}_I.$ As in Eq.~(\ref{Quantum-Interference}), Eq.~(\ref{Higher-order}) is to be interpreted as an inequality of the functions defined on the right and left. See Appendix~(\ref{Higher-order-App}) for an in-depth discussion of Eq.~(\ref{Higher-order}).

Motivated by Eq.~(\ref{Higher-order}), we wish to determine if particular phase transformations give rise to higher-order interference. The defining feature of phase transformations is that they leave the statistics of effects with support on single paths -- that is, effects in $\bigcup_i\mathcal{E}_{\{i\}}$ --  invariant. The natural generalisation of this is to consider transformations that not only leave the statistics of effects on single paths invariant, but also superposition effects. This motivates the following definition.
\begin{definition}
A transformation $T$ is \emph{$n$-undetectable} if:
$(e|T=(e|, \  \forall (e|\in\bigcup_{I:|I|\leq n}\mathcal{E}_I$.
\end{definition}
Together with its natural converse.
\begin{definition}
A transformation $T$ is \emph{$m$-detectable} if there exists $(e|\in\bigcup_{I:|I|\leq m}\mathcal{E}_I$, such that $(e|T\neq(e|$.
\end{definition}
We can now link higher-order interference to certain types of phase transformations, which we call \emph{higher-order phases}.
\begin{theorem} \label{Higher-Phase}
A transformation $T$ that is $n$ detectable and $n\mathrm{-}1$ undetectable implies the existence of $n$th-order interference.
\end{theorem}
\begin{proof}
Choose $|s)$ and $(e|$ such that $T$ is detected. It is then clear that the left hand side of Eq.~(\ref{Higher-order}) is dependent on $T$, whilst -- due to undetectability -- the right hand side is not. They are thus distinct functions.
\end{proof}
In our quantum example, the phase transformation was $2$-detectable, but $1$-undetectable.

\section{Controlled transformations and a generalised phase kick-back} \label{Control2}
\begin{definition}
Given a set of pure and perfectly distinguishable states $\{|i)\}$ and a set of transformations $\{T_i\}$,  we define a controlled transformation $C\{T_i\}$ as:
\begin{equation}\label{Control}
\begin{tikzpicture}
	\begin{pgfonlayer}{nodelayer}
		\node [style=cpoint] (0) at (0, 1) {$i$};
		\node [style=none] (1) at (1, 1.5) {};
		\node [style=none] (2) at (2.5, 1.5) {};
		\node [style=none] (3) at (1, -0.75) {};
		\node [style=none] (4) at (2.5, -0.75) {};
		\node [style=none] (5) at (1, 1) {};
		\node [style=none] (6) at (2.5, 1) {};
		\node [style=none] (7) at (1, -0.25) {};
		\node [style=none] (8) at (2.5, -0.25) {};
		\node [style=cpoint] (9) at (0, -0.25) {$\sigma$};
		\node [style=none] (10) at (3.5, -0.25) {};
		\node [style=none] (11) at (3.5, 1) {};
		\node [style=none] (12) at (1.75, 1) {$C$};
		\node [style=none] (13) at (5, 0.5) {$=$};
		\node [style=cpoint] (14) at (6.5, 1) {$i$};
		\node [style=none] (15) at (9.5, 1) {};
		\node [style=cpoint] (16) at (6.5, -0.25) {$\sigma$};
		\node [style={small box}] (17) at (8, -0.25) {$T_i$};
		\node [style=none] (18) at (9.5, -0.25) {};
		\node [style=none] (19) at (1.75, -0.25) {$\{T_i\}$};
		\node [style=none] (20) at (11.5, -0) {$\forall i, |\sigma)$};
		\node [style=none] (21) at (12.5, -0) {};
	\end{pgfonlayer}
	\begin{pgfonlayer}{edgelayer}
		\draw (0) to (5.center);
		\draw (9) to (7.center);
		\draw (6.center) to (11.center);
		\draw (8.center) to (10.center);
		\draw (1.center) to (3.center);
		\draw (3.center) to (4.center);
		\draw (4.center) to (2.center);
		\draw (2.center) to (1.center);
		\draw (16) to (17);
		\draw (17) to (18.center);
		\draw (14) to (15.center);
	\end{pgfonlayer}
\end{tikzpicture} \end{equation}
The top system and lower systems are referred to as the \emph{control} and \emph{target} respectively.
\end{definition}
Note that classical controlled transformations -- where the control is measured and conditioned on the outcome a transformation is applied to the target -- exist in any causal theory \cite{Pavia1} with sufficient distinguishable states. However, such transformations are in general not reversible and do not offer an advantage over classical computation \cite{Rev}. Moreover, the existence of reversible controlled transformations appears to be a rare property of operational theories \cite{Rev}. The following states that in theories satisfying our assumptions, there exist reversible controlled transformations. The proof is in contained in Appendix~(\ref{CT}).

\begin{theorem} \label{Reversible-Control}
In any theory satisfying i) causality, ii) purification, iii) strong symmetry, there exists a \emph{reversible} controlled transformation for all sets of reversible transformations $\{T_i\}$.
\end{theorem}

Moreover, the following theorem states that any controlled transformation in such theories `preserves superpositions'. Where `superposition' is meant in the sense of Sec.~(\ref{int}) part (ii) and `preserves superposition' means that the probability of detecting the system in each path of the superposition is preserved by the transformation. See Appendix~(\ref{SuperpositionProof}) for the proof.
\begin{lemma} Superpositions are preserved on the control input:
\begin{equation} \label{SupPres}
\begin{tikzpicture}
	\begin{pgfonlayer}{nodelayer}
		\node [style=none] (0) at (1, 1.5) {};
		\node [style=none] (1) at (2.5, 1.5) {};
		\node [style=none] (2) at (1, -0.75) {};
		\node [style=none] (3) at (2.5, -0.75) {};
		\node [style=none] (4) at (1, 1) {};
		\node [style=none] (5) at (2.5, 1) {};
		\node [style=none] (6) at (1, -0.25) {};
		\node [style=none] (7) at (2.5, -0.25) {};
		\node [style=cocpoint] (8) at (3.5, 1) {$i$};
		\node [style=none] (9) at (1.75, 1) {$C$};
		\node [style=none] (10) at (5, 0.5) {$=$};
		\node [style=none] (11) at (0, 1) {};
		\node [style=none] (12) at (6.5, 1) {};
		\node [style=cocpoint] (13) at (9.5, 1) {$i$};
		\node [style={small box}] (14) at (8, -0.25) {$T_i$};
		\node [style=cpoint] (15) at (0, -0.25) {$\sigma$};
		\node [style=none] (16) at (3.5, -0.25) {};
		\node [style=cpoint] (17) at (6.5, -0.25) {$\sigma$};
		\node [style=none] (18) at (9.5, -0.25) {};
		\node [style=none] (19) at (1.75, -0.25) {$\{T_i\}$};
		\node [style=none] (20) at (11.5, -0) {$\forall i,\ |\sigma)$};
		\node [style=none] (21) at (12.25, -0) {};
	\end{pgfonlayer}
	\begin{pgfonlayer}{edgelayer}
		\draw (5.center) to (8);
		\draw (0.center) to (2.center);
		\draw (2.center) to (3.center);
		\draw (3.center) to (1.center);
		\draw (1.center) to (0.center);
		\draw (11.center) to (4.center);
		\draw (12.center) to (13);
		\draw (15) to (6.center);
		\draw (7.center) to (16.center);
		\draw (17) to (14);
		\draw (14) to (18.center);
	\end{pgfonlayer}
\end{tikzpicture} \end{equation}
where $\{(i|\}$ is the measurement that perfectly distinguishes the control states $\{|i)\}$ \footnote{In theories satisfying strong symmetry the measurement $\{(i|\}$ is unique up to normalisation, see Appendix~(\ref{Sym}).}.
\end{lemma}

Every controlled transformation in quantum theory has a \emph{phase kick-back} mechanism \cite{Nielsen}. Such mechanisms form a vital component of most quantum algorithms. We now show the existence of a \emph{generalised} phase kick-back mechanism in any theory satisfying our assumptions.
\begin{lemma} \label{Generalised-Kick-Back}
Given an $|s)$ such that $T_i|s)=|s), \ \forall{i}$, there exists a reversible transformation $Q_s$ such that
\begin{equation}\label{KB}
\begin{tikzpicture}
	\begin{pgfonlayer}{nodelayer}
		\node [style=cpoint] (0) at (0, -0.25) {$s$};
		\node [style=none] (1) at (1.5, 1) {$C$};
		\node [style=cpoint] (2) at (6, 1) {$\sigma$};
		\node [style=none] (3) at (2.25, 1.5) {};
		\node [style=none] (4) at (9, -0.25) {};
		\node [style={small box}] (5) at (7.5, 1) {$Q_s$};
		\node [style=none] (6) at (0.75, -0.75) {};
		\node [style=none] (7) at (10.5, -0) {$\forall |\sigma)$};
		\node [style=none] (8) at (0.75, 1.5) {};
		\node [style=none] (9) at (0.75, -0.25) {};
		\node [style=none] (10) at (1.5, -0.25) {$\{T_i\}$};
		\node [style=none] (11) at (4.5, 0.5) {$=$};
		\node [style=none] (12) at (9, 1) {};
		\node [style=none] (13) at (3, 1) {};
		\node [style=none] (14) at (2.25, -0.75) {};
		\node [style=none] (15) at (0.75, 1) {};
		\node [style=cpoint] (16) at (6, -0.25) {$s$};
		\node [style=none] (17) at (2.25, -0.25) {};
		\node [style=none] (18) at (3, -0.25) {};
		\node [style=none] (19) at (2.25, 1) {};
		\node [style=cpoint] (20) at (0, 1) {$\sigma$};
		\node [style=none] (21) at (11, -0) {};
	\end{pgfonlayer}
	\begin{pgfonlayer}{edgelayer}
		\draw (20) to (15.center);
		\draw (8.center) to (6.center);
		\draw (6.center) to (14.center);
		\draw (14.center) to (3.center);
		\draw (3.center) to (8.center);
		\draw (0) to (9.center);
		\draw (17.center) to (18.center);
		\draw (19.center) to (13.center);
		\draw (2) to (5);
		\draw (16) to (4.center);
		\draw (5) to (12.center);
	\end{pgfonlayer}
\end{tikzpicture} \end{equation}
Moreover, $Q_s$ is phase transformation:
\[\begin{tikzpicture}
	\begin{pgfonlayer}{nodelayer}
		\node [style={small box}] (0) at (1, -0) {$Q_s$};
		\node [style=none] (1) at (5.5, -0) {};
		\node [style=none] (2) at (0, -0) {};
		\node [style=none] (3) at (8, -0) {$\forall i$};
		\node [style=none] (4) at (4, -0) {$=$};
		\node [style=cocpoint] (5) at (2.5, -0) {$i$};
		\node [style=cocpoint] (6) at (6.5, -0) {$i$};
		\node [style=none] (7) at (8.5, -0) {};
	\end{pgfonlayer}
	\begin{pgfonlayer}{edgelayer}
		\draw (5) to (0);
		\draw (0) to (2.center);
		\draw (6) to (1.center);
	\end{pgfonlayer}
\end{tikzpicture} \]
\end{lemma}
\proof
\[\begin{tikzpicture}
	\begin{pgfonlayer}{nodelayer}
		\node [style=none] (0) at (9.25, -0.75) {};
		\node [style=none] (1) at (9.25, 1) {};
		\node [style=none] (2) at (7.75, 1) {};
		\node [style=none] (3) at (4.5, 0.5) {$=$};
		\node [style=none] (4) at (2.25, -0.25) {};
		\node [style=trace] (5) at (14.75, 1) {};
		\node [style=none] (6) at (7.75, -0.25) {};
		\node [style=none] (7) at (7.75, 1.5) {};
		\node [style=cpoint] (8) at (7, -0.25) {$s$};
		\node [style=none] (9) at (1.5, 1) {$C$};
		\node [style=none] (10) at (0, 1) {};
		\node [style=none] (11) at (12.75, 1) {};
		\node [style=none] (12) at (0.75, 1.5) {};
		\node [style=none] (13) at (10, -0.25) {};
		\node [style=none] (14) at (7, 1) {};
		\node [style=none] (15) at (7.75, -0.75) {};
		\node [style=trace] (16) at (3, 1) {};
		\node [style=none] (17) at (2.25, -0.75) {};
		\node [style=cocpoint] (18) at (10, 1) {$i$};
		\node [style=none] (19) at (9.25, -0.25) {};
		\node [style=cpoint] (20) at (0, -0.25) {$s$};
		\node [style=none] (21) at (0.75, -0.75) {};
		\node [style=none] (22) at (11.5, 0.5) {$=$};
		\node [style=none] (23) at (9.25, 1.5) {};
		\node [style=none] (24) at (3, -0.25) {};
		\node [style=none] (25) at (2.25, 1.5) {};
		\node [style=none] (26) at (0.75, 1) {};
		\node [style=none] (27) at (14.75, -0) {};
		\node [style=none] (28) at (5.75, 0.5) {$\sum_i$};
		\node [style=none] (29) at (8.5, -0.25) {$\{T_i\}$};
		\node [style=none] (30) at (8.5, 1) {$C$};
		\node [style=none] (31) at (1.5, -0.25) {$\{T_i\}$};
		\node [style=none] (32) at (2.25, 1) {};
		\node [style=none] (33) at (0.75, -0.25) {};
		\node [style=cpoint] (34) at (12.75, -0) {$s$};
		\node [style=none] (35) at (15.25, -0) {};
	\end{pgfonlayer}
	\begin{pgfonlayer}{edgelayer}
		\draw (10.center) to (26.center);
		\draw (12.center) to (21.center);
		\draw (21.center) to (17.center);
		\draw (17.center) to (25.center);
		\draw (25.center) to (12.center);
		\draw (20) to (33.center);
		\draw (4.center) to (24.center);
		\draw (32.center) to (16);
		\draw (14.center) to (2.center);
		\draw (7.center) to (15.center);
		\draw (15.center) to (0.center);
		\draw (0.center) to (23.center);
		\draw (23.center) to (7.center);
		\draw (8) to (6.center);
		\draw (19.center) to (13.center);
		\draw (1.center) to (18);
		\draw (11.center) to (5);
		\draw (34) to (27.center);
	\end{pgfonlayer}
\end{tikzpicture} \]
The first equality follows from causality and the second from Eq.~(\ref{SupPres}) and the definition of $|s)$. Eq.~(\ref{Dilation}) then implies the existence of a reversible $Q_s$ such that:
\[\begin{tikzpicture}
	\begin{pgfonlayer}{nodelayer}
		\node [style=cpoint] (0) at (0, -0.25) {$s$};
		\node [style=none] (1) at (0.75, 1) {};
		\node [style=none] (2) at (8.75, -0.25) {};
		\node [style=cpoint] (3) at (0, 1) {$\sigma$};
		\node [style=none] (4) at (2.25, -0.25) {};
		\node [style=cpoint] (5) at (6.25, 1) {$\sigma$};
		\node [style=none] (6) at (0.75, 1.5) {};
		\node [style=cpoint] (7) at (6.25, -0.25) {$s$};
		\node [style=none] (8) at (1.5, -0.25) {$\{T_i\}$};
		\node [style=none] (9) at (3, 1) {};
		\node [style=none] (10) at (0.75, -0.75) {};
		\node [style=none] (11) at (8.75, 1) {};
		\node [style=none] (12) at (2.25, -0.75) {};
		\node [style=none] (13) at (1.5, 1) {$C$};
		\node [style=none] (14) at (3, -0.25) {};
		\node [style=none] (15) at (2.25, 1) {};
		\node [style={small box}] (16) at (7.5, 1) {$Q_s$};
		\node [style=none] (17) at (0.75, -0.25) {};
		\node [style=none] (18) at (10.5, -0) {$\forall |\sigma)$};
		\node [style=none] (19) at (4.5, 0.5) {$=$};
		\node [style=none] (20) at (2.25, 1.5) {};
		\node [style=none] (21) at (11.25, -0) {};
	\end{pgfonlayer}
	\begin{pgfonlayer}{edgelayer}
		\draw (3) to (1.center);
		\draw (6.center) to (10.center);
		\draw (10.center) to (12.center);
		\draw (12.center) to (20.center);
		\draw (20.center) to (6.center);
		\draw (0) to (17.center);
		\draw (4.center) to (14.center);
		\draw (15.center) to (9.center);
		\draw (5) to (16);
		\draw (7) to (2.center);
		\draw (16) to (11.center);
	\end{pgfonlayer}
\end{tikzpicture} \]
Note that $Q_s$ depends on both the controlled transformation and the joint eigenstate $|s)$. Note that:
\[\begin{tikzpicture}
	\begin{pgfonlayer}{nodelayer}
		\node [style=none] (0) at (0.75, 2.75) {};
		\node [style={small box}] (1) at (8.25, -0.75) {$T_i$};
		\node [style=none] (2) at (1.5, 3.25) {$\{T_i\}$};
		\node [style=trace] (3) at (9.5, -0.75) {};
		\node [style=none] (4) at (1.5, 4.5) {$C$};
		\node [style=none] (5) at (4.5, 1.75) {$\diageq$};
		\node [style=cpoint] (6) at (6.75, 3) {$s$};
		\node [style=none] (7) at (2.25, 4.5) {};
		\node [style=none] (8) at (0.75, 5) {};
		\node [style=cpoint] (9) at (6.75, -0.75) {$s$};
		\node [style=cpoint] (10) at (0, 4.5) {$\sigma$};
		\node [style=trace] (11) at (9.5, 3) {};
		\node [style=cocpoint] (12) at (9.5, 0.75) {$i$};
		\node [style=none] (13) at (0.75, 4.5) {};
		\node [style=none] (14) at (4.5, 4) {$=$};
		\node [style=none] (15) at (2.25, 3.25) {};
		\node [style=cpoint] (16) at (6.75, 0.75) {$\sigma$};
		\node [style=cocpoint] (17) at (3, 4.5) {$i$};
		\node [style=trace] (18) at (3, 3.25) {};
		\node [style=cpoint] (19) at (0, 3.25) {$s$};
		\node [style=none] (20) at (0.75, 3.25) {};
		\node [style=none] (21) at (2.25, 5) {};
		\node [style=cocpoint] (22) at (9.5, 4.25) {$i$};
		\node [style=cpoint] (23) at (6.75, 4.25) {$\sigma$};
		\node [style={small box}] (24) at (8, 4.25) {$Q_s$};
		\node [style=none] (25) at (2.25, 2.75) {};
		\node [style=none] (26) at (10.5, -0) {};
		\node [style=none] (27) at (11, -0) {$\forall \sigma$};
	\end{pgfonlayer}
	\begin{pgfonlayer}{edgelayer}
		\draw (10) to (13.center);
		\draw (8.center) to (0.center);
		\draw (0.center) to (25.center);
		\draw (25.center) to (21.center);
		\draw (21.center) to (8.center);
		\draw (19) to (20.center);
		\draw (15.center) to (18);
		\draw (7.center) to (17);
		\draw (23) to (24);
		\draw (6) to (11);
		\draw (24) to (22);
		\draw (9) to (1);
		\draw (16) to (12);
		\draw (1) to (3);
	\end{pgfonlayer}
\end{tikzpicture}\]
Causality -- via state normalisation -- then gives:
\[\begin{tikzpicture}
	\begin{pgfonlayer}{nodelayer}
		\node [style=none] (0) at (5.5, -0) {};
		\node [style=cocpoint] (1) at (6.75, -0) {$i$};
		\node [style=none] (2) at (0, -0) {};
		\node [style={small box}] (3) at (1.25, -0) {$Q_s$};
		\node [style=none] (4) at (8.5, -0) {$\forall i$};
		\node [style=cocpoint] (5) at (2.5, -0) {$i$};
		\node [style=none] (6) at (4, -0) {$=$};
		\node [style=none] (7) at (9.25, -0) {};
	\end{pgfonlayer}
	\begin{pgfonlayer}{edgelayer}
		\draw (2.center) to (3);
		\draw (3) to (5);
		\draw (0.center) to (1);
	\end{pgfonlayer}
\end{tikzpicture} \]
\endproof

In quantum theory, it is possible to achieve any phase transformation via a kick-back mechanism. However, Thm.~(\ref{Generalised-Kick-Back}) only implies the existence of at least one phase that can be `kicked-back'. We now show that all phases arise via the generalised mechanism. Consider the set of pure and perfectly distinguishable states $\{|s_i)\}$ and let $\{T_i\}$ be elements of their phase group, i.e. $T_i|s_j)=|s_j),$ $\forall i,j$. Construct the controlled transformation $C\{T_i\}$. The designation of control and target for $C\{T_i\}$ is symmetric:
\[\begin{tikzpicture}
	\begin{pgfonlayer}{nodelayer}
		\node [style=none] (0) at (0.75, 1.5) {};
		\node [style=none] (1) at (0.75, -0.25) {};
		\node [style=none] (2) at (0.75, -0.75) {};
		\node [style=none] (3) at (2.25, -0.75) {};
		\node [style=none] (4) at (2.25, -0.25) {};
		\node [style=none] (5) at (2.25, 1.5) {};
		\node [style=none] (6) at (2.25, 2) {};
		\node [style=none] (7) at (0.75, 2) {};
		\node [style=cpoint] (8) at (0, 1.5) {$\sigma$};
		\node [style=cpoint] (9) at (0, -0.25) {$s_i$};
		\node [style=none] (10) at (3, -0.25) {};
		\node [style=none] (11) at (3, 1.5) {};
		\node [style=none] (12) at (1.5, 1.5) {$C$};
		\node [style=none] (13) at (4.25, 0.75) {$=$};
		\node [style={small box}] (14) at (7, 1.5) {$Q_i$};
		\node [style=none] (15) at (8.25, -0.25) {};
		\node [style=cpoint] (16) at (5.75, -0.25) {$s_i$};
		\node [style=cpoint] (17) at (5.75, 1.5) {$\sigma$};
		\node [style=none] (18) at (8.25, 1.5) {};
		\node [style=none] (19) at (9.5, 0.75) {$:=$};
		\node [style=none] (20) at (1.5, -0.25) {$\{T_i\}$};
		\node [style=none] (21) at (12, -0.75) {};
		\node [style=none] (22) at (13.5, -0.75) {};
		\node [style=cpoint] (23) at (11.25, 1.5) {$\sigma$};
		\node [style=none] (24) at (14.25, -0.25) {};
		\node [style=none] (25) at (14.25, 1.5) {};
		\node [style=none] (26) at (12, 1.5) {};
		\node [style=none] (27) at (13.5, 2) {};
		\node [style=none] (28) at (12.75, -0.25) {$C$};
		\node [style=none] (29) at (13.5, 1.5) {};
		\node [style=cpoint] (30) at (11.25, -0.25) {$s_i$};
		\node [style=none] (31) at (13.5, -0.25) {};
		\node [style=none] (32) at (12.75, 1.5) {$\{Q_i\}$};
		\node [style=none] (33) at (12, 2) {};
		\node [style=none] (34) at (12, -0.25) {};
		\node [style=none] (35) at (15.75, -0) {$\forall |\sigma)$};
		\node [style=none] (36) at (16.5, -0) {};
	\end{pgfonlayer}
	\begin{pgfonlayer}{edgelayer}
		\draw (8) to (0.center);
		\draw (7.center) to (2.center);
		\draw (2.center) to (3.center);
		\draw (3.center) to (6.center);
		\draw (6.center) to (7.center);
		\draw (9) to (1.center);
		\draw (4.center) to (10.center);
		\draw (5.center) to (11.center);
		\draw (17) to (14);
		\draw (14) to (18.center);
		\draw (16) to (15.center);
		\draw (23) to (26.center);
		\draw (33.center) to (21.center);
		\draw (21.center) to (22.center);
		\draw (22.center) to (27.center);
		\draw (27.center) to (33.center);
		\draw (30) to (34.center);
		\draw (31.center) to (24.center);
		\draw (29.center) to (25.center);
	\end{pgfonlayer}
\end{tikzpicture} \]
Thus any $C\{T_i\}$ is control-target symmetric if the $T_i$ are elements of a phase group. The transformations on the target are given by the kicked-back phases, $\{Q_i\}$. Given an arbitrary $W_i$, construct the transformation $\{W_i\}C$ and note that via control-target symmetry it is equivalent to $C\{G_i\}$, for some $\{G_i\}$. The controlled transformation $C\{G_i\}$ thus gives rise to the kicked-back phase $W_i$ and we have:
\begin{theorem} \label{ALL}
Every phase transformation can arise via a generalised phase kick-back mechanism
\end{theorem}

\subsection{Particle exchange experiments} \label{Exchange}

Dahlsten et al. \cite{Oscar} have shown that there is a close connection between particle exchange statistics and the phase group in operational theories. We use the framework and results presented in this paper to expand upon and formalise these connections. Motivated by the quantum case, place a pair of indistinguishable particles in superposition by inputting them to an interferometer, as shown in the following diagram. 
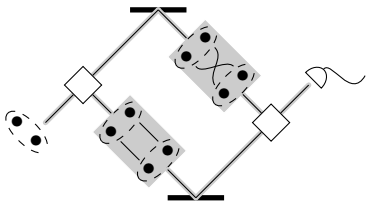
\[\begin{tikzpicture}
	\begin{pgfonlayer}{nodelayer}
		\node [style=none] (0) at (0, 1.25) {};
		\node [style=di, fill=white] (1) at (0.9999999, 2.25) {};
		\node [style=none] (2) at (7, 2.25) {};
		\node [style=di, fill=white] (3) at (6, 1.25) { };
		\node [style=none] (4) at (2.25, 4.25) {};
		\node [style=none] (5) at (3.75, 4.25) {};
		\node [style=none] (6) at (3.25, -0.7500001) {};
		\node [style=none] (7) at (4.75, -0.7500001) {};
		\node [style=mopoint, rotate=135, xscale=0.75, yscale=1] (8) at (7.25, 2.5) {};
		\node [style=none] (9) at (8.5, 2.5) {};
		\node [style=none] (10) at (-0.7499999, 1.25) {$\bullet$};
		\node [style=none] (11) at (-0.2499999, 0.7499999) {$\bullet$};
		\node [style=none] (12) at (-1, 1.5) {};
		\node [style=none] (13) at (0, 0.5000001) {};
		\node [style=none] (14) at (3, 4.25) {};
		\node [style=none] (15) at (3.75, 3.5) {};
		\node [style=none] (16) at (5.25, 2) {};
		\node [style=none] (17) at (4.25, 3.5) {$\bullet$};
		\node [style=none] (18) at (3.75, 3) {$\bullet$};
		\node [style=none] (19) at (4.5, 3.75) {};
		\node [style=none] (20) at (3.5, 2.75) {};
		\node [style=none] (21) at (4.5, 1.75) {};
		\node [style=none] (22) at (4.75, 2) {$\bullet$};
		\node [style=none] (23) at (5.25, 2.5) {$\bullet$};
		\node [style=none] (24) at (5.5, 2.75) {};
		\node [style=none] (25) at (4, 2.75) {};
		\node [style=none] (26) at (4.5, 2.25) {};
		\node [style=none] (27) at (5, 2.75) {};
		\node [style=none] (28) at (4.5, 3.25) {};
		\node [style=none] (29) at (1.5, 0.7500001) {};
		\node [style=none] (30) at (2.75, -0) {$\bullet$};
		\node [style=none] (31) at (1.75, 1) {$\bullet$};
		\node [style=none] (32) at (3, 0.7500001) {};
		\node [style=none] (33) at (2.5, 0.2500001) {};
		\node [style=none] (34) at (3.25, 0.4999999) {$\bullet$};
		\node [style=none] (35) at (2.5, -0.2500001) {};
		\node [style=none] (36) at (2.5, 1.75) {};
		\node [style=none] (37) at (3.5, 0.7500001) {};
		\node [style=none] (38) at (2.25, 1.5) {$\bullet$};
		\node [style=none] (39) at (4, -0.7500001) {};
		\node [style=none] (40) at (1.75, 1.5) {};
		\node [style=none] (41) at (3.25, -0) {};
		\node [style=none] (42) at (2, 0.7500001) {};
		\node [style=none] (43) at (2.5, 1.25) {};
		\node [style=none] (44) at (2.25, 2) {};
		\node [style=none] (45) at (1.25, 1) {};
		\node [style=none] (46) at (2.75, -0.4999999) {};
		\node [style=none] (47) at (3.75, 0.4999999) {};
		\node [style=none] (48) at (4.25, 4) {};
		\node [style=none] (49) at (3.25, 3) {};
		\node [style=none] (50) at (4.75, 1.5) {};
		\node [style=none] (51) at (5.75, 2.5) {};
	\end{pgfonlayer}
	\begin{pgfonlayer}{edgelayer}
		\filldraw[fill=gray!40, draw=white!0] (44.center) -- (47.center) -- (46.center) -- (45.center) -- cycle;
		\filldraw[fill=gray!40, draw =white!0] (49.center) to (50.center) to (51.center) to (48.center) -- cycle;
		\draw [style=none, in=-120, out=30, looseness=1.25] (8) to (9.center);
		\draw [style=mirror] (4.center) to (5.center);
		\draw [style=mirror] (6.center) to (7.center);
		\draw [style={dashed edge}, bend right=90, looseness=0.75] (12.center) to (13.center);
		\draw [style={dashed edge}, bend left=90, looseness=0.75] (12.center) to (13.center);
		\draw [style=particlePath] (0.center) to (1);
		\draw [style=particlePath] (1) to (14.center);
		\draw [style=particlePath] (14.center) to (15.center);
		\draw [style=particlePath] (16.center) to (3);
		\draw [style={dashed edge}, bend right=90, looseness=0.75] (20.center) to (19.center);
		\draw [style={dashed edge}, bend left=90, looseness=0.75] (20.center) to (19.center);
		\draw [style={dashed edge}, bend right=90, looseness=0.75] (21.center) to (24.center);
		\draw [style={dashed edge}, bend left=90, looseness=0.75] (21.center) to (24.center);
		\draw [in=135, out=-60, looseness=0.75] (25.center) to (27.center);
		\draw [in=120, out=-45, looseness=1.00] (28.center) to (26.center);
		\draw [style={dashed edge}, bend right=90, looseness=0.75] (29.center) to (36.center);
		\draw [style={dashed edge}, bend left=90, looseness=0.75] (29.center) to (36.center);
		\draw [style={dashed edge}, bend right=90, looseness=0.75] (35.center) to (37.center);
		\draw [style={dashed edge}, bend left=90, looseness=0.75] (35.center) to (37.center);
		\draw [style=particlePath] (39.center) to (3);
		\draw [style=particlePath] (3) to (2.center);
		\draw [style=particlePath] (1) to (40.center);
		\draw [style=particlePath] (41.center) to (39.center);
		\draw (42.center) to (33.center);
		\draw (43.center) to (32.center);
		\draw (0.center) to (1);
		\draw (1) to (14.center);
		\draw (14.center) to (15.center);
		\draw (16.center) to (3);
		\draw (2.center) to (3);
		\draw (3) to (39.center);
		\draw (39.center) to (41.center);
		\draw (40.center) to (1);
	\end{pgfonlayer}
\end{tikzpicture} \]
On the upper path the two particles are swapped using some operation `$S$', whilst on the lower path they are left invariant -- that is, the identity operation $\mathds{1}$ is applied. The entire physical set-up is described by a bipartite state, one partition of which corresponds to the state of the particles, $|s)\in\Omega_{P'cles}$, and the other to the `which path' information embodied in the interferometer, $|s')\in\Omega_{Path}$. The entire scenario thus takes place in the state space $\Omega_{Path}\otimes \Omega_{P'cle}$. In the quantum case, the phase transformation generated by this procedure corresponds to the type of indistinguishable particle employed in the experiment.

The whole experiment can be described via a controlled transformation, with path information as the control and particle state as the target. Via Thm.~(\ref{Reversible-Control}), such an experiment exists in theories satisfying our assumptions. Applying operation $S$ to the particle state corresponds to swapping a pair of indistinguishable particles and so must leave the statistics of any measurement invariant. Therefore $S|s)=|s)$, where $|s)$ is the initial particle state.
\[\begin{tikzpicture}
	\begin{pgfonlayer}{nodelayer}
		\node [style=none] (0) at (1.5, 1) {};
		\node [style=none] (1) at (2.75, 1) {};
		\node [style=none] (2) at (4.75, 1) {};
		\node [style=none] (3) at (5.5, 1) {};
		\node [style=none] (4) at (2.75, 1.5) {};
		\node [style=none] (5) at (4.75, 1.5) {};
		\node [style=cpoint] (6) at (2, -0.25) {$s$};
		\node [style=none] (7) at (2.75, -0.25) {};
		\node [style=none] (8) at (2.75, -0.75) {};
		\node [style=none] (9) at (4.75, -0.75) {};
		\node [style=none] (10) at (4.75, -0.25) {};
		\node [style=none] (11) at (5.5, -0.25) {};
		\node [style=none] (12) at (3.75, -0.25) {$\{\mathds{1},S\}$};
		\node [style=none] (13) at (3.75, 1) {$C$};
		\node [style=none] (14) at (0.25, 1) {$\Omega_{Path}$};
		\node [style=none] (15) at (0.25, -0.25) {$\Omega_{P'cle}$};
		\node [style=none] (16) at (7, 0.5) {where};
		\node [style=cpoint] (17) at (9, 0.5) {$s$};
		\node [style={small box}] (18) at (10.25, 0.5) {$S$};
		\node [style=none] (19) at (11.25, 0.5) {};
		\node [style=none] (20) at (12.25, 0.5) {$=$};
		\node [style=cpoint] (21) at (13.25, 0.5) {$s$};
		\node [style=none] (22) at (14.25, 0.5) {};
		\node [style=none] (23) at (14.75, 0.5) {};
	\end{pgfonlayer}
	\begin{pgfonlayer}{edgelayer}
		\draw (0.center) to (1.center);
		\draw (6) to (7.center);
		\draw (10.center) to (11.center);
		\draw (2.center) to (3.center);
		\draw (4.center) to (5.center);
		\draw (5.center) to (9.center);
		\draw (9.center) to (8.center);
		\draw (8.center) to (4.center);
		\draw (17) to (18);
		\draw (18) to (19.center);
		\draw (21) to (22.center);
	\end{pgfonlayer}
\end{tikzpicture}  \]
Thm.~(\ref{Generalised-Kick-Back}) tells us that the above diagram corresponds to a kicked-back phase on the control system, as in Eq.~(\ref{KB}).
Thus, to every particle type, there exists a corresponding phase transformation, which was the connection discussed in \cite{Oscar}. But, in an arbitrary theory, the converse is not necessarily true. The quantum phase group is $U(1)$ and, fixing its representation to be $\{e^{i\theta}\}$, bosons kick-back the transformation corresponding to $\theta=0$, fermions  $\theta=\pi$ and anyons any arbitrary $\theta$. Thus, to every particle type in quantum theory there is an associated phase, and vice versa.

To generalise this to theories satisfying our three assumptions, we must connect the operational description of these theories to the more physical notion of particles. Towards this end, we make the following two assumptions:
\begin{enumerate}
\item Every operational state $|s)$ corresponds to the state of some collection of indistinguishable particles,
\item Every transformation that leaves the operational state $|s)$ invariant corresponds to a (possibly trivial) permutation of the collection of indistinguishable particles.
\end{enumerate}
Given the above, Thm.~(\ref{ALL}) tells us that to every phase transformation there exists a corresponding particle type. Therefore, to each higher-order phase -- described in Thm.~(\ref{Higher-Phase}) -- there is associated a particle type that should be observable through a generalisation of the above experiment.

Consider $|s)$, which corresponds to the state of some collection of indistinguishable particles, and a permutation operation $\pi$ which leaves $|s)$ invariant. Note that, for a given permutation, there may be multiple topologically distinct ways of performing it, particularly in two dimensions or topologically non-trivial spaces. Now consider the $n$-path experiment, illustrated below, where on each path some distinct permutation operation $\pi_i,\ i=1,\dots,n,$ takes place.
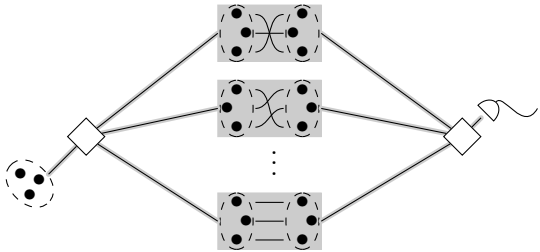
\[\begin{tikzpicture}
	\begin{pgfonlayer}{nodelayer}
		\node [style=none] (0) at (0, 1.25) {};
		\node [style=di, fill=white] (1) at (0.9999999, 2.25) {};
		\node [style=none] (2) at (11.5, 2.75) {};
		\node [style=di, fill=white] (3) at (11, 2.25) { };
		\node [style=mopoint, rotate=135, xscale=0.75, yscale=1] (4) at (11.75, 3) {};
		\node [style=none] (5) at (13, 3) {};
		\node [style=none] (6) at (-0.7499999, 1.25) {$\bullet$};
		\node [style=none] (7) at (-0.2500001, 1.1) {$\bullet$};
		\node [style=none] (8) at (-1, 1.5) {};
		\node [style=none] (9) at (0, 0.5000001) {};
		\node [style=none] (10) at (4.5, 5) {};
		\node [style=none] (11) at (7.25, 5) {};
		\node [style=none] (12) at (5, 5.5) {$\bullet$};
		\node [style=none] (13) at (5.25, 5) {$\bullet$};
		\node [style=none] (14) at (5, 5.75) {};
		\node [style=none] (15) at (5, 4.25) {};
		\node [style=none] (16) at (6.75, 4.25) {};
		\node [style=none] (17) at (6.75, 4.5) {$\bullet$};
		\node [style=none] (18) at (6.75, 5.5) {$\bullet$};
		\node [style=none] (19) at (6.75, 5.75) {};
		\node [style=none] (20) at (5.5, 4.5) {};
		\node [style=none] (21) at (6.25, 4.5) {};
		\node [style=none] (22) at (6.25, 5.5) {};
		\node [style=none] (23) at (5.5, 5.5) {};
		\node [style=none] (24) at (5, -0.75) {};
		\node [style=none] (25) at (6.75, -0.5000002) {$\bullet$};
		\node [style=none] (26) at (5, -0.5000002) {$\bullet$};
		\node [style=none] (27) at (6.25, 0.5000002) {};
		\node [style=none] (28) at (6.25, -0.5000002) {};
		\node [style=none] (29) at (6.75, 0.5000002) {$\bullet$};
		\node [style=none] (30) at (6.75, -0.75) {};
		\node [style=none] (31) at (5, 0.75) {};
		\node [style=none] (32) at (6.75, 0.75) {};
		\node [style=none] (33) at (5, 0.5000002) {$\bullet$};
		\node [style=none] (34) at (4.5, -0) {};
		\node [style=none] (35) at (5.5, -0.5000002) {};
		\node [style=none] (36) at (5.5, 0.5000002) {};
		\node [style=none] (37) at (4.5, 0.75) {};
		\node [style=none] (38) at (4.5, -0.75) {};
		\node [style=none] (39) at (7.25, -0.75) {};
		\node [style=none] (40) at (7.25, 0.75) {};
		\node [style=none] (41) at (4.5, 5.75) {};
		\node [style=none] (42) at (4.5, 4.25) {};
		\node [style=none] (43) at (7.25, 4.25) {};
		\node [style=none] (44) at (7.25, 5.75) {};
		\node [style=none] (45) at (7.25, -0) {};
		\node [style=none] (46) at (-0.5000001, 0.7) {$\bullet$};
		\node [style=none] (47) at (5, 2.25) {};
		\node [style=none] (48) at (5.5, 2.5) {};
		\node [style=none] (49) at (7.25, 3.75) {};
		\node [style=none] (50) at (6.75, 3.5) {$\bullet$};
		\node [style=none] (51) at (4.5, 3.75) {};
		\node [style=none] (52) at (5, 2.5) {$\bullet$};
		\node [style=none] (53) at (6.25, 2.5) {};
		\node [style=none] (54) at (6.75, 2.25) {};
		\node [style=none] (55) at (4.5, 3) {};
		\node [style=none] (56) at (7.25, 3) {};
		\node [style=none] (57) at (4.5, 2.25) {};
		\node [style=none] (58) at (6.75, 2.5) {$\bullet$};
		\node [style=none] (59) at (5, 3.5) {$\bullet$};
		\node [style=none] (60) at (6.25, 3) {};
		\node [style=none] (61) at (6.75, 3.75) {};
		\node [style=none] (62) at (7.25, 2.25) {};
		\node [style=none] (63) at (5, 3.75) {};
		\node [style=none] (64) at (5.5, 3.5) {};
		\node [style=none] (65) at (4.5, 3) {};
		\node [style=none] (66) at (5, 4.5) {$\bullet$};
		\node [style=none] (67) at (6.5, 5) {$\bullet$};
		\node [style=none] (68) at (4.75, 3) {$\bullet$};
		\node [style=none] (69) at (7, 3) {$\bullet$};
		\node [style=none] (70) at (5.25, -0) {$\bullet$};
		\node [style=none] (71) at (7, -0) {$\bullet$};
		\node [style=none] (72) at (5.5, 5) {};
		\node [style=none] (73) at (6.25, 5) {};
		\node [style=none] (74) at (5.5, -0) {};
		\node [style=none] (75) at (6.25, -0) {};
		\node [style=none] (76) at (5.5, 3) {};
		\node [style=none] (77) at (6.25, 3.5) {};
		\node [style=none] (78) at (6, 1.75) {$\vdots$};
	\end{pgfonlayer}
	\begin{pgfonlayer}{edgelayer}
		\filldraw [fill=gray!40, draw = gray!40] (37.center) to (40.center) to (39.center) to (38.center) -- cycle;
		\filldraw [fill=gray!40, draw = gray!40] (57.center) to (62.center) to (49.center) to (51.center) -- cycle;
		\filldraw [fill=gray!40, draw = gray!40] (41.center) to (44.center) to (43.center) to (42.center) -- cycle;
		\draw [style=none, in=-120, out=30, looseness=1.25] (4) to (5.center);
		\draw [style={dashed edge}, bend right=90, looseness=1.25] (8.center) to (9.center);
		\draw [style={dashed edge}, bend left=90, looseness=1.25] (8.center) to (9.center);
		\draw [style=particlePath] (0.center) to (1);
		\draw [style=particlePath] (1) to (10.center);
		\draw [style=particlePath] (11.center) to (3);
		\draw [style={dashed edge}, bend right=90, looseness=1.00] (15.center) to (14.center);
		\draw [style={dashed edge}, bend left=90, looseness=1.00] (15.center) to (14.center);
		\draw [style={dashed edge}, bend right=90, looseness=1.00] (16.center) to (19.center);
		\draw [style={dashed edge}, bend left=90, looseness=1.00] (16.center) to (19.center);
		\draw [in=180, out=4, looseness=1.00] (20.center) to (22.center);
		\draw [in=180, out=0, looseness=1.00] (23.center) to (21.center);
		\draw [style={dashed edge}, bend right=90, looseness=1.00] (24.center) to (31.center);
		\draw [style={dashed edge}, bend left=90, looseness=1.00] (24.center) to (31.center);
		\draw [style={dashed edge}, bend right=90, looseness=1.00] (30.center) to (32.center);
		\draw [style={dashed edge}, bend left=90, looseness=1.00] (30.center) to (32.center);
		\draw [style=particlePath] (3) to (2.center);
		\draw [style=particlePath] (1) to (34.center);
		\draw (35.center) to (28.center);
		\draw (36.center) to (27.center);
		\draw (0.center) to (1);
		\draw (1) to (10.center);
		\draw (11.center) to (3);
		\draw (2.center) to (3);
		\draw (34.center) to (1);
		\draw [style=particlePath] (45.center) to (3);
		\draw (45.center) to (3);
		\draw [style={dashed edge}, bend right=90, looseness=1.00] (47.center) to (63.center);
		\draw [style={dashed edge}, bend left=90, looseness=1.00] (47.center) to (63.center);
		\draw [style={dashed edge}, bend right=90, looseness=1.00] (54.center) to (61.center);
		\draw [style={dashed edge}, bend left=90, looseness=1.00] (54.center) to (61.center);
		\draw [in=180, out=4, looseness=1.00] (48.center) to (60.center);
		\draw [in=180, out=0, looseness=1.00] (64.center) to (53.center);
		\draw [style=particlePath] (1) to (55.center);
		\draw [style=particlePath] (56.center) to (3);
		\draw (1) to (55.center);
		\draw (56.center) to (3);
		\draw (72.center) to (73.center);
		\draw (74.center) to (75.center);
		\draw [in=180, out=0, looseness=1.00] (76.center) to (77.center);
	\end{pgfonlayer}
\end{tikzpicture}  \]
This can be described by a controlled transformation $C\{\pi_i\}$. Given the above two assumptions, any $n$th-order phase -- if they exist in the operational theory -- can be kicked-back by such an experiment.

Recall that $n$th-order phases are $n$-detectable, but $n\mathrm{-}1$-undetectable. That is, the action of such phases cannot be detected by any effect with support on less than $n$ paths. Which is in stark contrast to quantum, or $2$nd-order, phases, which can always be detected by an effect with support on two paths. Thus, permutations of particles, whose type all correspond to an $n$th-order phase, can only be detected by recombining \emph{all} paths in an $n$-path experiment. In some sense then, $n$th-order phases encode holistic information about all paths in an $n$-path experiment.

\subsection{Computational oracles} \label{oracle}

Oracles play a vital role in quantum computing, forming the basis of most known speed-ups over classical computation \cite{Nielsen}. Despite their importance, defining a general notion of oracle -- that reduces to the standard notion in the quantum case -- in operationally-defined theories has proven difficult \cite{LB-2014}. A particular example of a quantum oracle is the following controlled unitary:
\begin{equation} \label{Quantum-Oracle}
U_f= \ket{0}\bra{0}\otimes Z^{f(0)} + \ket{1}\bra{1}\otimes Z^{f(1)},
\end{equation}
with $Z$ a Pauli matrix, $f:\{0,1\}\rightarrow\{0,1\}$ a function encoding some decision problem and $Z^{0}:=\mathbb{I}$. The quantum phase kick-back for $U_f$ amounts to
\begin{equation} \label{Quantum-kick-back}
U_f=\mathbb{I}\otimes\ket{0}\bra{0}+Z^{f(0)\oplus f(1)}\otimes\ket{1}\bra{1}.
\end{equation}
One can see that inputting $\ket{+}\ket{1}$ and measuring the first qubit in the $\{\ket{+},\ket{-}\}$ basis reveals the value of $f(0)\oplus f(1)$ in a single query of the oracle -- a feat impossible on a classical computer \cite{Nielsen}.

The results of Thm.~(\ref{Reversible-Control}) provide a way to define computational oracles in any theory satisfying our three assumptions. An oracle in such theories corresponds to a reversible controlled transformation \footnote{There could be many distinct transformations that have the same behaviour on a set of control states. As long as one fixes which transformation corresponds to the oracle, this is not a problem.} where the set of transformations $\{T_{i,f(i)}\}$ being controlled depend on a function $f:\{i\}\rightarrow\{0,1\}$ encoding a decision problem of interest. As the transformations $T_{i,f(i)}$ depend on the value of $f(i)$, so does the controlled transformation and the kicked-back phase. That is, in theories with a non-trivial phase group, the phase kick-back of an oracle encodes information about the value $f(i)$ for all $i$. In such theories, there is thus a non-zero probability of extracting such global information. Non-trivial interference behaviour can thus be seen as a general resource for non-classical computation.

In the quantum case, there is a limit to how much global information one can obtain in a single oracle query. In the situation where $f:\{0,\dots, n\mathrm{-}1\}\rightarrow\{0,1\}$ a quantum oracle can only extract the value of $f(i)\oplus f(j)$, for some $i,j$, in a single query without error \cite{Nielsen}. Can theories with higher-order interference reliably extract more global information about $f$ -- without error -- in a single query? The results of Sec.~(\ref{Exchange}) appear to suggest that $n$th-order phases encode information about all paths in an $n$-path experiment, as opposed to $2$nd-order, or quantum, phases which only encode information about at most two paths. Based on this fact -- that higher-order phases encode more holistic information that quantum ones --
and the result of Thm.~(\ref{ALL}), we conjecture that theories with higher-order interference can solve problems intractable on a quantum computer. To prove such a conjecture however, a concrete example of such a theory is needed. While there are partial examples in the literature \cite{DensityCube, Thesis}, none of these are complete \cite{LS-2015}. Answering such a conjecture in the affirmative is thus not yet possible, although some evidence was provided in \cite{LS-2015}.

\section{Conclusion}

The key result of this paper was to provide a set of physical principles that are sufficient for the existence of reversible controlled transformations. Such transformations are central to our understanding of quantum computing, information processing and thermodynamics. Moreover, these were shown to guarantee the existence of a generalised phase kick-back mechanism, which, in the quantum case, forms a fundamental component of almost all algorithms. These physical principles are defining characteristics of information: independence of encoding medium; propagation from present to future; and conservation at a fundamental level. It would therefore be surprising if these principles were not necessary primers for information processing. These results provide the tools for an exploration of the structure of computational algorithms -- and how they connect to physical principles -- in operational theories.

We developed a framework that connects higher-order interference and phase transformations, generalising the intimate connection between phase and interference witnessed in quantum theory. These `higher-order' phases are accessible via our generalised kick-back mechanism. Given two assumptions which connect the operational theory to a physical description of particles, these higher-order phases were shown to give rise to exotic particle types. Additionally, using the controlled transformations to define an oracle model of computation, we conjectured that these higher-order phases may allow for the solution of problems intractable even on a quantum computer. Computational problems that may be susceptible to efficient solution by generalised phase kick-back include the $n$-collision problem, and the non-abelian hidden subgroup problem. Discovering that higher-order interference leads to  `unreasonable' computational power may provide a reason `why' quantum theory is limited in its interference behaviour -- in the same way that implausible communication complexity is thought to limit quantum non-locality \cite{PR-van-Dam}.

In Sec.~(\ref{Exchange}) it was shown that to observe the exotic particle types corresponding to higher-order phases, there must be distinct ways to swap particles. As we live in a topologically trivial three dimensional space, there is only one topologically distinct way to swap point particles. This can either be seen as evidence of \emph{why} quantum theory is limited to only second-order interference, or evidence that such particle types must have non-trivial structure, similar to toroidal anyons \cite{torons} -- which are constructed from a solenoid ring with an attached charge -- or closed strings \cite{Barton}.

Finally, reference \cite{work} has shown that thermodynamic work can be extracted from quantum coherences -- $2$nd-order phases in our language. This raises this the question of whether one can extract work more efficiently using higher-order phases? If such efficiencies are in contention with thermodynamic principles this could provide a reason `why' quantum theory has limited interference. Initial investigations into formulating a consistent thermodynamics in operational theories have been reported in \cite{Thermo, Thermo1,Thermo2}. The framework and results presented here may therefore have implications for thermodynamics, information processing and how each arises in a unified manner from physical principles.

\section*{Acknowledgements}
The authors thank Howard Barnum, Oscar Dahlsten, Terry Rudolph, Jon Barrett and Matty Hoban for useful discussions. Matty Hoban is also thanked for proof reading a draft of the current paper. The authors also thank Carlo Maria Scandolo for a careful reading of the appendices in a previous draft of the current paper. This work was supported by EPSRC through the Controlled Quantum Dynamics Centre for Doctoral Training and the Oxford Department of Computer Science. CML also acknowledges funding from University College, Oxford.

\appendix

\section{General results following from causality purification and strong symmetry}
\subsection{Uniqueness of distinguishing measurement} \label{Sym}
Strong symmetry (together with the no restriction hypothesis, which says that all mathematically well-defined measurements are physical) implies that, given any set of pure and perfectly distinguishable states $\{|i)\}$, there exists a unique measurement $\{(j|\}$ such that, \[(i|j)=\delta_{ij}.\] See \cite{MU, Higher-order-reconstruction} for details. Moreover if there is a set $\{(e_j|\}$ such that $(e_j|i)=\alpha_j \delta_{ij}$ then, \[(e_j|=\alpha_j(j|.\]
\vspace{-1cm}
\subsection{Existence of a maximally mixed state}
Purification implies that there is a unique \emph{completely mixed} state $\maxMix$ defined by, \[T\maxMix=\maxMix,\quad \forall T\in\mathcal{R}\]
Any state is a `refinement' of this state. See \cite{Pavia1, Pavia2} for details.

\subsection{Purification of the maximally mixed state is dynamically faithful}
Purification implies that there exists a state $|\psi)$ that purifies the completely mixed state:
\[\begin{tikzpicture}
	\begin{pgfonlayer}{nodelayer}
		\node [style=none] (0) at (1, -0.25) {};
		\node [style=none] (1) at (1, 1.25) {};
		\node [style=none] (2) at (1, 1.75) {};
		\node [style=none] (3) at (1, -0.75) {};
		\node [style=trace] (4) at (2.5, -0.25) {};
		\node [style=none] (5) at (2.5, 1.25) {};
		\node [style=none] (6) at (1, -0.75) {};
		\node [style=none] (7) at (0.5, 0.5) {$\psi$};
		\node [style=none] (8) at (3.75, 0.75) {$=$};
		\node [style=traceState] (9) at (5, 0.75) {};
		\node [style=none] (10) at (6.5, 0.75) {};
		\node [style=none] (11) at (7.25, -0) {};
	\end{pgfonlayer}
	\begin{pgfonlayer}{edgelayer}
		\draw [bend right=90, looseness=1.25] (2.center) to (3.center);
		\draw (2.center) to (3.center);
		\draw (1.center) to (5.center);
		\draw (0.center) to (4);
		\draw (9) to (10.center);
	\end{pgfonlayer}
\end{tikzpicture} \]
This is unique up to reversible transformation. We denote a particular choice of this purification as,
\[\begin{tikzpicture}
	\begin{pgfonlayer}{nodelayer}
		\node [style=none] (0) at (4.75, -0.25) {};
		\node [style=none] (1) at (4.75, 1.25) {};
		\node [style=none] (2) at (4.75, 1.75) {};
		\node [style=none] (3) at (4.75, -0.75) {};
		\node [style=none] (4) at (6.25, -0.25) {};
		\node [style=none] (5) at (6.25, 1.25) {};
		\node [style=none] (6) at (4.75, -0.75) {};
		\node [style=none] (7) at (4.25, 0.5) {$\psi$};
		\node [style=none] (8) at (2.5, 0.5) {$:=$};
		\node [style=none] (9) at (0.75, 1.25) {};
		\node [style=none] (10) at (1.25, -0.25) {};
		\node [style=none] (11) at (0.75, -0.25) {};
		\node [style=none] (12) at (1.25, 1.25) {};
		\node [style=none] (13) at (7, -0) {};
	\end{pgfonlayer}
	\begin{pgfonlayer}{edgelayer}
		\draw [bend right=90, looseness=1.25] (2.center) to (3.center);
		\draw (2.center) to (3.center);
		\draw (1.center) to (5.center);
		\draw (0.center) to (4.center);
		\draw (9.center) to (12.center);
		\draw (11.center) to (10.center);
		\draw [bend right=90, looseness=1.75] (9.center) to (11.center);
	\end{pgfonlayer}
\end{tikzpicture} \]
Purifications of the completely mixed state are called \emph{dynamically faithful} states \cite{Pavia1,Pavia2} and satisfy the following important condition \cite{Pavia1,Pavia2}:
\[\begin{tikzpicture}
	\begin{pgfonlayer}{nodelayer}
		\node [style=none] (0) at (5.5, 7.25) {$=$};
		\node [style=none] (1) at (0.7499999, 7) {};
		\node [style=none] (2) at (1.25, 5.25) {};
		\node [style=none] (3) at (0.7499997, 5.25) {};
		\node [style=none] (4) at (1.25, 7) {};
		\node [style=none] (5) at (1.25, 9) {};
		\node [style=none] (6) at (1.25, 6) {};
		\node [style=none] (7) at (2.75, 6) {};
		\node [style=none] (8) at (2.75, 9) {};
		\node [style=none] (9) at (2, 7.5) {$T$};
		\node [style=none] (10) at (1.25, 8.5) {};
		\node [style=none] (11) at (1.25, 7.5) {};
		\node [style=none] (12) at (0, 8.5) {};
		\node [style=none] (13) at (3.75, 7) {};
		\node [style=none] (14) at (3.75, 8.5) {};
		\node [style=none] (15) at (2.75, 8.5) {};
		\node [style=none] (16) at (2.75, 7) {};
		\node [style=none] (17) at (3.75, 5.25) {};
		\node [style=none] (18) at (8.25, 6) {};
		\node [style=none] (19) at (8.25, 5.25) {};
		\node [style=none] (20) at (9.75, 7) {};
		\node [style=none] (21) at (8.25, 8.5) {};
		\node [style=none] (22) at (8.25, 9) {};
		\node [style=none] (23) at (9.75, 8.5) {};
		\node [style=none] (24) at (10.75, 5.25) {};
		\node [style=none] (25) at (9.75, 6) {};
		\node [style=none] (26) at (9.000001, 7.5) {$T'$};
		\node [style=none] (27) at (6.999999, 8.5) {};
		\node [style=none] (28) at (7.75, 7) {};
		\node [style=none] (29) at (8.25, 7.5) {};
		\node [style=none] (30) at (10.75, 8.5) {};
		\node [style=none] (31) at (9.75, 9) {};
		\node [style=none] (32) at (8.25, 7) {};
		\node [style=none] (33) at (7.75, 5.25) {};
		\node [style=none] (34) at (10.75, 7) {};
		\node [style=none] (35) at (2, 3.5) {$\implies$};
		\node [style=none] (36) at (3.75, -0.7499997) {};
		\node [style=cpoint] (37) at (2.749999, -0) {$\sigma$};
		\node [style=none] (38) at (13.25, 1.75) {};
		\node [style=none] (39) at (11.5, 0.7499999) {$T'$};
		\node [style=none] (40) at (12.25, 0.25) {};
		\node [style=none] (41) at (5.25, 1.75) {};
		\node [style=none] (42) at (5.25, 2.25) {};
		\node [style=none] (43) at (12.25, -0.7499999) {};
		\node [style=none] (44) at (3.75, 1.75) {};
		\node [style=none] (45) at (3.75, -0) {};
		\node [style=none] (46) at (12.25, 2.25) {};
		\node [style=none] (47) at (9.500001, 1.75) {};
		\node [style=none] (48) at (3.75, 0.7499997) {};
		\node [style=none] (49) at (10.75, 1.75) {};
		\node [style=none] (50) at (5.25, 0.25) {};
		\node [style=none] (51) at (2.5, 1.75) {};
		\node [style=none] (52) at (3.75, 2.25) {};
		\node [style=none] (53) at (8, 0.5000001) {$=$};
		\node [style=none] (54) at (13.25, 0.2499996) {};
		\node [style=none] (55) at (10.75, -0) {};
		\node [style=cpoint] (56) at (9.750001, -0) {$\sigma$};
		\node [style=none] (57) at (6.25, 1.75) {};
		\node [style=none] (58) at (4.5, 0.7499999) {$T$};
		\node [style=none] (59) at (12.25, 1.75) {};
		\node [style=none] (60) at (10.75, -0.7499997) {};
		\node [style=none] (61) at (6.249999, 0.2499996) {};
		\node [style=none] (62) at (5.25, -0.7499999) {};
		\node [style=none] (63) at (10.75, 0.7499997) {};
		\node [style=none] (64) at (10.75, 2.25) {};
		\node [style=none] (65) at (15, -0) {$\forall \sigma$};
		\node [style=none] (66) at (2.749999, -0) {};
		\node [style=none] (67) at (15.75, -0) {};
	\end{pgfonlayer}
	\begin{pgfonlayer}{edgelayer}
		\draw (1.center) to (4.center);
		\draw (3.center) to (2.center);
		\draw [bend right=90, looseness=1.75] (1.center) to (3.center);
		\draw (5.center) to (6.center);
		\draw (6.center) to (7.center);
		\draw (7.center) to (8.center);
		\draw (8.center) to (5.center);
		\draw (12.center) to (10.center);
		\draw (15.center) to (14.center);
		\draw (16.center) to (13.center);
		\draw (2.center) to (17.center);
		\draw (28.center) to (32.center);
		\draw (33.center) to (19.center);
		\draw [bend right=90, looseness=1.75] (28.center) to (33.center);
		\draw (22.center) to (18.center);
		\draw (18.center) to (25.center);
		\draw (25.center) to (31.center);
		\draw (31.center) to (22.center);
		\draw (27.center) to (21.center);
		\draw (23.center) to (30.center);
		\draw (20.center) to (34.center);
		\draw (19.center) to (24.center);
		\draw (37) to (45.center);
		\draw (52.center) to (36.center);
		\draw (36.center) to (62.center);
		\draw (62.center) to (42.center);
		\draw (42.center) to (52.center);
		\draw (51.center) to (44.center);
		\draw (41.center) to (57.center);
		\draw (50.center) to (61.center);
		\draw (64.center) to (60.center);
		\draw (60.center) to (43.center);
		\draw (43.center) to (46.center);
		\draw (46.center) to (64.center);
		\draw (47.center) to (49.center);
		\draw (56) to (55.center);
		\draw (59.center) to (38.center);
		\draw (40.center) to (54.center);
	\end{pgfonlayer}
\end{tikzpicture}  \]

\section{Existence of controlled transformations\label{CT}} 
Recalling that in this work we assume that the composite of pure states is pure, we can define two sets of pure and perfectly distinguishable states:
\[\mathcal{B}_1 :=
\left\{\begin{tikzpicture}
	\begin{pgfonlayer}{nodelayer}
		\node [style=none] (0) at (1, -0) {};
		\node [style=none] (1) at (2, -0) {};
		\node [style=none] (2) at (2, -1) {};
		\node [style=none] (3) at (2, 1) {};
		\node [style=none] (4) at (2, -1) {};
		\node [style=cpoint] (5) at (0.4999999, 0.9999999) {$i$};
		\node [style=none] (6) at (1, -1) {};
	\end{pgfonlayer}
	\begin{pgfonlayer}{edgelayer}
		\draw [bend right=90, looseness=2.25] (0.center) to (6.center);
		\draw (5) to (3.center);
		\draw (6.center) to (2.center);
		\draw (0.center) to (1.center);
	\end{pgfonlayer}
\end{tikzpicture}\right\}  ,\quad \text{and}\quad\mathcal{B}_2 := \left\{\begin{tikzpicture}
	\begin{pgfonlayer}{nodelayer}
		\node [style=none] (0) at (1, -0) {};
		\node [style={small box}] (1) at (1.75, -0) {$T_i$};
		\node [style=none] (2) at (2.75, -0) {};
		\node [style=none] (3) at (2.75, -0.9999999) {};
		\node [style=none] (4) at (2.75, 0.9999999) {};
		\node [style=none] (5) at (2.75, -0.9999999) {};
		\node [style=cpoint] (6) at (0.4999999, 0.9999999) {$i$};
		\node [style=none] (7) at (1, -1) {};
	\end{pgfonlayer}
	\begin{pgfonlayer}{edgelayer}
		\draw [bend right=90, looseness=2.25] (0.center) to (7.center);
		\draw (6) to (4.center);
		\draw (7.center) to (3.center);
		\draw (0.center) to (1);
		\draw (1) to (2.center);
	\end{pgfonlayer}
\end{tikzpicture}\right\}  .\]

Strong symmetry implies that there exists a reversible transformation between these two sets, $T:\mathcal{B}_1\to \mathcal{B}_2$.
\[\begin{tikzpicture}
	\begin{pgfonlayer}{nodelayer}
		\node [style=none] (0) at (7.75, 0.7499997) {};
		\node [style=none] (1) at (1.75, 0.7499997) {$T$};
		\node [style=none] (2) at (2.25, -0.7499997) {};
		\node [style=none] (3) at (1.25, 1.75) {};
		\node [style=none] (4) at (5.25, 0.7499997) {$=$};
		\node [style=none] (5) at (1.25, 0.7499997) {};
		\node [style=none] (6) at (3.000001, -0.2499996) {};
		\node [style=none] (7) at (1.25, 2.25) {};
		\node [style=none] (8) at (1.25, -0.2499996) {};
		\node [style=none] (9) at (0.7499997, -0.2499996) {};
		\node [style=none] (10) at (3.000001, 0.7499997) {};
		\node [style=none] (11) at (3.000001, -0.2499996) {};
		\node [style=none] (12) at (2.25, -0.2499996) {};
		\node [style={small box}] (13) at (8.499999, 0.7499997) {$T_i$};
		\node [style=none] (14) at (9.499999, 0.7499997) {};
		\node [style=none] (15) at (9.499999, -0.2499996) {};
		\node [style=none] (16) at (9.499999, 1.75) {};
		\node [style=cpoint] (17) at (0.2499996, 1.75) {$i$};
		\node [style=none] (18) at (0.7499997, 0.7499997) {};
		\node [style=none] (19) at (9.499999, -0.2499996) {};
		\node [style=none] (20) at (3.000001, 1.75) {};
		\node [style=none] (21) at (2.25, 2.25) {};
		\node [style=cpoint] (22) at (7.25, 1.75) {$i$};
		\node [style=none] (23) at (2.25, 1.75) {};
		\node [style=none] (24) at (2.25, 0.7499997) {};
		\node [style=none] (25) at (7.75, -0.2499996) {};
		\node [style=none] (26) at (1.25, -0.7499997) {};
		\node [style=none] (27) at (10, -0) {};
	\end{pgfonlayer}
	\begin{pgfonlayer}{edgelayer}
		\draw (7.center) to (21.center);
		\draw (21.center) to (2.center);
		\draw (2.center) to (26.center);
		\draw [bend right=90, looseness=2.25] (18.center) to (9.center);
		\draw (7.center) to (26.center);
		\draw (17) to (3.center);
		\draw (18.center) to (5.center);
		\draw (9.center) to (8.center);
		\draw (12.center) to (11.center);
		\draw (24.center) to (10.center);
		\draw (23.center) to (20.center);
		\draw [bend right=90, looseness=2.25] (0.center) to (25.center);
		\draw (22) to (16.center);
		\draw (25.center) to (15.center);
		\draw (0.center) to (13);
		\draw (13) to (14.center);
	\end{pgfonlayer}
\end{tikzpicture} \]
This result, together with the existence of dynamically faithful states, will be used to show the existence of a reversible controlled transformation $C\{T_i\}$ for an arbitrary set of reversible transformations $\{T_i\}$.

\begin{lemma}`Superposition preservation'\label{SP}
\[\begin{tikzpicture}
	\begin{pgfonlayer}{nodelayer}
		\node [style=none] (0) at (7.75, 0.7499997) {};
		\node [style=none] (1) at (1.75, 0.7499997) {$T$};
		\node [style=none] (2) at (2.25, -0.7499997) {};
		\node [style=none] (3) at (1.25, 1.75) {};
		\node [style=none] (4) at (5.25, 0.7499997) {$=$};
		\node [style=none] (5) at (1.25, 0.7499997) {};
		\node [style=none] (6) at (3.000001, -0.2499996) {};
		\node [style=none] (7) at (1.25, 2.25) {};
		\node [style=none] (8) at (1.25, -0.2499996) {};
		\node [style=none] (9) at (0.7499997, -0.2499996) {};
		\node [style=none] (10) at (3.000001, 0.7499997) {};
		\node [style=none] (11) at (3.000001, -0.2499996) {};
		\node [style=none] (12) at (2.25, -0.2499996) {};
		\node [style={small box}] (13) at (8.499999, 0.7499997) {$T_i$};
		\node [style=none] (14) at (9.499999, 0.7499997) {};
		\node [style=none] (15) at (9.499999, -0.2499996) {};
		\node [style=cocpoint] (16) at (9.499999, 1.75) {$i$};
		\node [style=none] (17) at (0.2499996, 1.75) {};
		\node [style=none] (18) at (0.7499997, 0.7499997) {};
		\node [style=none] (19) at (9.499999, -0.2499996) {};
		\node [style=cocpoint] (20) at (3.000001, 1.75) {$i$};
		\node [style=none] (21) at (2.25, 2.25) {};
		\node [style=none] (22) at (7.25, 1.75) {};
		\node [style=none] (23) at (2.25, 1.75) {};
		\node [style=none] (24) at (2.25, 0.7499997) {};
		\node [style=none] (25) at (7.75, -0.2499996) {};
		\node [style=none] (26) at (1.25, -0.7499997) {};
		\node [style=none] (27) at (10, -0) {};
	\end{pgfonlayer}
	\begin{pgfonlayer}{edgelayer}
		\draw (7.center) to (21.center);
		\draw (21.center) to (2.center);
		\draw (2.center) to (26.center);
		\draw [bend right=90, looseness=2.25] (18.center) to (9.center);
		\draw (7.center) to (26.center);
		\draw (17.center) to (3.center);
		\draw (18.center) to (5.center);
		\draw (9.center) to (8.center);
		\draw (12.center) to (11.center);
		\draw (24.center) to (10.center);
		\draw (23.center) to (20);
		\draw [bend right=90, looseness=2.25] (0.center) to (25.center);
		\draw (22.center) to (16);
		\draw (25.center) to (15.center);
		\draw (0.center) to (13);
		\draw (13) to (14.center);
	\end{pgfonlayer}
\end{tikzpicture} \]
\end{lemma}
\proof
Firstly we prove a weaker condition which is superposition preservation for pure local effects,
\[\begin{tikzpicture}
	\begin{pgfonlayer}{nodelayer}
		\node [style=cpoint] (0) at (0, 10.25) {$i$};
		\node [style=none] (1) at (2.250001, 10.25) {};
		\node [style=none] (2) at (2.250001, 10.75) {};
		\node [style=none] (3) at (3.75, 10.75) {};
		\node [style=none] (4) at (2.250001, 8.5) {};
		\node [style=none] (5) at (2.250001, 6.5) {};
		\node [style=none] (6) at (3.75, 6.5) {};
		\node [style=none] (7) at (3.75, 8.5) {};
		\node [style=none] (8) at (3.75, 10.25) {};
		\node [style=cocpoint] (9) at (4.75, 10.25) {$j$};
		\node [style=trace] (10) at (4.75, 8.5) {};
		\node [style=none] (11) at (3, 8.5) {$T$};
		\node [style=none] (12) at (0.7499999, 11.25) {};
		\node [style=none] (13) at (0.7499999, 6) {};
		\node [style=none] (14) at (5.499999, 6) {};
		\node [style=none] (15) at (5.499999, 11.25) {};
		\node [style=none] (16) at (6.499999, 9.25) {$=$};
		\node [style=none] (17) at (12.25, 9) {$\delta_{ij}$};
		\node [style={small box}] (18) at (9.25, 9.75) {$T_j$};
		\node [style=trace] (19) at (10.5, 9.75) {};
		\node [style=none] (20) at (4.5, 5) {Strong symmetry $\implies$};
		\node [style=trace] (21) at (4.75, 7) {};
		\node [style=none] (22) at (8.500001, 8.25) {};
		\node [style=none] (23) at (8.500001, 9.75) {};
		\node [style=trace] (24) at (10.5, 8.25) {};
		\node [style=none] (25) at (13.25, 1.5) {};
		\node [style=trace] (26) at (14.5, -0) {};
		\node [style=none] (27) at (13.25, -0) {};
		\node [style=none] (28) at (10.5, 1.75) {$=$};
		\node [style=trace] (29) at (14.5, 1.5) {};
		\node [style=cocpoint] (30) at (14.5, 3) {$j$};
		\node [style=none] (31) at (12, 3) {};
		\node [style=none] (32) at (2.250001, 7) {};
		\node [style=none] (33) at (3.75, 7) {};
		\node [style=none] (34) at (4.25, 4) {};
		\node [style=none] (35) at (7.25, -0.7499999) {};
		\node [style=cocpoint] (36) at (8.25, 3) {$j$};
		\node [style=none] (37) at (5.75, 3.5) {};
		\node [style=none] (38) at (7.25, 3) {};
		\node [style=trace] (39) at (8.25, -0.25) {};
		\node [style=none] (40) at (9.000001, 4) {};
		\node [style=none] (41) at (5.75, -0.25) {};
		\node [style=none] (42) at (7.25, -0.25) {};
		\node [style=none] (43) at (5.75, -0.7499999) {};
		\node [style=none] (44) at (5.75, 1.25) {};
		\node [style=none] (45) at (7.25, 1.25) {};
		\node [style=trace] (46) at (8.25, 1.25) {};
		\node [style=none] (47) at (4.25, -1.25) {};
		\node [style=none] (48) at (5.75, 3) {};
		\node [style=none] (49) at (7.25, 3.5) {};
		\node [style=none] (50) at (9.000001, -1.25) {};
		\node [style=none] (51) at (6.499999, 1.25) {$T$};
		\node [style=none] (52) at (3.5, 3) {};
		\node [style=none] (53) at (16.5, -0) {};
		\node [style=none] (54) at (14, 9) {$=$};
		\node [style=none] (55) at (15.75, 9) {$\delta_{ij}$};
	\end{pgfonlayer}
	\begin{pgfonlayer}{edgelayer}
		\draw (2.center) to (5.center);
		\draw (5.center) to (6.center);
		\draw (6.center) to (3.center);
		\draw (3.center) to (2.center);
		\draw (0) to (1.center);
		\draw (7.center) to (10);
		\draw (8.center) to (9);
		\draw [style={thick gray dashed edge}] (12.center) to (15.center);
		\draw [style={thick gray dashed edge}] (15.center) to (14.center);
		\draw [style={thick gray dashed edge}] (14.center) to (13.center);
		\draw [style={thick gray dashed edge}] (13.center) to (12.center);
		\draw (18) to (19);
		\draw [bend right=90, looseness=1.50] (23.center) to (22.center);
		\draw (23.center) to (18);
		\draw (22.center) to (24);
		\draw [bend right=90, looseness=1.50] (25.center) to (27.center);
		\draw (27.center) to (26);
		\draw (31.center) to (30);
		\draw (33.center) to (21);
		\draw [bend left=90, looseness=2.00] (32.center) to (4.center);
		\draw (37.center) to (43.center);
		\draw (43.center) to (35.center);
		\draw (35.center) to (49.center);
		\draw (49.center) to (37.center);
		\draw (45.center) to (46);
		\draw (38.center) to (36);
		\draw [style={thick gray dashed edge}] (34.center) to (40.center);
		\draw [style={thick gray dashed edge}] (40.center) to (50.center);
		\draw [style={thick gray dashed edge}] (50.center) to (47.center);
		\draw [style={thick gray dashed edge}] (47.center) to (34.center);
		\draw (42.center) to (39);
		\draw [bend left=90, looseness=2.00] (41.center) to (44.center);
		\draw (52.center) to (48.center);
		\draw (25.center) to (29);
	\end{pgfonlayer}
\end{tikzpicture} \]
the implication follows from the uniqueness of the maximally distinguishing measurement up to normalisation. Then purification implies,
\[\begin{tikzpicture}
	\begin{pgfonlayer}{nodelayer}
		\node [style=none] (0) at (6.75, 2.5) {};
		\node [style=none] (1) at (7.25, 0.25) {};
		\node [style=none] (2) at (5.5, 2) {$=$};
		\node [style=cocpoint] (3) at (8.75, 2.5) {$i$};
		\node [style=none] (4) at (7.25, 1.25) {};
		\node [style=none] (5) at (7.25, 1.25) {};
		\node [style=none] (6) at (7.25, 1.75) {};
		\node [style=none] (7) at (8.25, 1.75) {};
		\node [style=none] (8) at (8.25, -0.25) {};
		\node [style=none] (9) at (7.25, -0.25) {};
		\node [style=none] (10) at (7.25, 0.25) {};
		\node [style=none] (11) at (7.25, 1.25) {};
		\node [style=none] (12) at (8.25, 1.25) {};
		\node [style=none] (13) at (8.25, 0.25) {};
		\node [style=none] (14) at (9.000001, 1.25) {};
		\node [style=none] (15) at (9.000001, 0.25) {};
		\node [style=none] (16) at (7.75, 0.7499999) {$T'_i$};
		\node [style=none] (17) at (0, 2.25) {};
		\node [style=none] (18) at (2.5, -0.7499999) {};
		\node [style=none] (19) at (2.5, 2.25) {};
		\node [style=none] (20) at (0.9999999, 2.75) {};
		\node [style=none] (21) at (0.9999999, 0.9999999) {};
		\node [style=none] (22) at (0.9999999, -0.25) {};
		\node [style=none] (23) at (2.5, 0.9999999) {};
		\node [style=none] (24) at (2.5, -0.25) {};
		\node [style=none] (25) at (0.9999999, 2.25) {};
		\node [style=none] (26) at (2.5, 2.75) {};
		\node [style=none] (27) at (0.9999999, -0.7499999) {};
		\node [style=cocpoint] (28) at (3.5, 2.25) {$i$};
		\node [style=none] (29) at (2.5, 2.25) {};
		\node [style=none] (30) at (3.5, 0.9999999) {};
		\node [style=none] (31) at (3.5, -0.25) {};
		\node [style=none] (32) at (1.75, 0.9999999) {$T$};
		\node [style=none] (33) at (9.75, -0) {};
	\end{pgfonlayer}
	\begin{pgfonlayer}{edgelayer}
		\draw [bend right=90, looseness=2.00] (4.center) to (1.center);
		\draw (0.center) to (3);
		\draw (4.center) to (5.center);
		\draw (6.center) to (9.center);
		\draw (9.center) to (8.center);
		\draw (8.center) to (7.center);
		\draw (7.center) to (6.center);
		\draw (12.center) to (14.center);
		\draw (13.center) to (15.center);
		\draw (20.center) to (27.center);
		\draw (27.center) to (18.center);
		\draw (18.center) to (26.center);
		\draw (26.center) to (20.center);
		\draw (17.center) to (25.center);
		\draw [bend right=90, looseness=2.00] (21.center) to (22.center);
		\draw (29.center) to (28);
		\draw (23.center) to (30.center);
		\draw (24.center) to (31.center);
	\end{pgfonlayer}
\end{tikzpicture} \]
Now consider,
\[\begin{tikzpicture}
	\begin{pgfonlayer}{nodelayer}
		\node [style=none] (0) at (0.9999999, 7) {};
		\node [style=none] (1) at (0.9999999, 8.25) {};
		\node [style=none] (2) at (2.5, 8.25) {};
		\node [style=none] (3) at (3.5, 8.25) {};
		\node [style=none] (4) at (2.5, 6.5) {};
		\node [style=none] (5) at (0.9999999, 6.5) {};
		\node [style=none] (6) at (0.9999999, 10) {};
		\node [style=cocpoint] (7) at (3.5, 9.5) {$i$};
		\node [style=none] (8) at (0.9999999, 9.5) {};
		\node [style=cpoint] (9) at (0, 9.5) {$i$};
		\node [style=none] (10) at (2.5, 9.5) {};
		\node [style=none] (11) at (2.5, 10) {};
		\node [style=none] (12) at (5, 8.5) {$=$};
		\node [style=none] (13) at (4.75, 5.5) {$\diageq$};
		\node [style=none] (14) at (7.75, 9.5) {};
		\node [style=none] (15) at (7.25, 7) {};
		\node [style=cpoint] (16) at (6.75, 9.5) {$i$};
		\node [style=none] (17) at (9.000001, 8) {};
		\node [style=cocpoint] (18) at (8.75, 9.5) {$i$};
		\node [style=none] (19) at (7.25, 8) {};
		\node [style=none] (20) at (8.25, 8) {};
		\node [style=none] (21) at (7.25, 8.5) {};
		\node [style=none] (22) at (8.25, 8.5) {};
		\node [style=none] (23) at (8.25, 6.5) {};
		\node [style=none] (24) at (7.25, 6.5) {};
		\node [style=none] (25) at (7.75, 7.5) {$T'_i$};
		\node [style=none] (26) at (9.000001, 7) {};
		\node [style=none] (27) at (8.25, 7) {};
		\node [style=cocpoint] (28) at (8.000001, 4.75) {$i$};
		\node [style=none] (29) at (7.25, 2.75) {};
		\node [style=none] (30) at (8.000001, 2.75) {};
		\node [style=none] (31) at (6.25, 2.75) {};
		\node [style=none] (32) at (6.999999, 4.75) {};
		\node [style=cpoint] (33) at (5.999999, 4.75) {$i$};
		\node [style={small box}] (34) at (6.999999, 3.75) {$T_i$};
		\node [style=none] (35) at (6.25, 3.75) {};
		\node [style=none] (36) at (8.000001, 3.75) {};
		\node [style=none] (37) at (0.9999999, 1.25) {$\implies$};
		\node [style=none] (38) at (3.5, 0.9999999) {};
		\node [style=none] (39) at (5, 0.9999999) {};
		\node [style=none] (40) at (5, -0) {};
		\node [style={small box}] (41) at (4, 0.9999999) {$T_i$};
		\node [style=none] (42) at (3.5, -0) {};
		\node [style=none] (43) at (6.499999, 0.4999999) {$=$};
		\node [style=none] (44) at (8.500001, 0.9999999) {};
		\node [style=none] (45) at (9.500001, 1.5) {};
		\node [style=none] (46) at (10.25, -0) {};
		\node [style=none] (47) at (9.500001, -0) {};
		\node [style=none] (48) at (9.500001, -0.4999999) {};
		\node [style=none] (49) at (8.500001, -0.4999999) {};
		\node [style=none] (50) at (8.500001, 1.5) {};
		\node [style=none] (51) at (9.500001, 0.9999999) {};
		\node [style=none] (52) at (9.000001, 0.4999999) {$T'_i$};
		\node [style=none] (53) at (10.25, 0.9999999) {};
		\node [style=none] (54) at (8.500001, -0) {};
		\node [style=none] (55) at (8.25, -0) {};
		\node [style=none] (56) at (8.25, 0.9999999) {};
		\node [style=none] (57) at (2.5, 7) {};
		\node [style=none] (58) at (3.5, 7) {};
		\node [style=none] (59) at (1.75, 8.25) {$T$};
		\node [style=none] (60) at (3.75, 0.9999999) {};
		\node [style=none] (61) at (3.5, -0) {};
		\node [style=none] (62) at (3.5, 0.9999999) {};
		\node [style=none] (63) at (8.25, -0) {};
		\node [style=none] (64) at (8.25, 0.9999999) {};
		\node [style=none] (65) at (11, -0) {};
	\end{pgfonlayer}
	\begin{pgfonlayer}{edgelayer}
		\draw (6.center) to (5.center);
		\draw (5.center) to (4.center);
		\draw (4.center) to (11.center);
		\draw (11.center) to (6.center);
		\draw (10.center) to (7);
		\draw (9) to (8.center);
		\draw [bend right=90, looseness=2.00] (1.center) to (0.center);
		\draw (2.center) to (3.center);
		\draw (14.center) to (18);
		\draw [bend right=90, looseness=2.00] (19.center) to (15.center);
		\draw (20.center) to (17.center);
		\draw (21.center) to (24.center);
		\draw (24.center) to (23.center);
		\draw (23.center) to (22.center);
		\draw (22.center) to (21.center);
		\draw (16) to (14.center);
		\draw (27.center) to (26.center);
		\draw (32.center) to (28);
		\draw [bend right=90, looseness=2.00] (35.center) to (31.center);
		\draw (34) to (36.center);
		\draw (33) to (32.center);
		\draw (29.center) to (30.center);
		\draw (31.center) to (29.center);
		\draw (35.center) to (34);
		\draw (41) to (39.center);
		\draw (42.center) to (40.center);
		\draw (38.center) to (41);
		\draw (51.center) to (53.center);
		\draw (50.center) to (49.center);
		\draw (49.center) to (48.center);
		\draw (48.center) to (45.center);
		\draw (45.center) to (50.center);
		\draw (47.center) to (46.center);
		\draw (56.center) to (44.center);
		\draw (55.center) to (54.center);
		\draw (57.center) to (58.center);
		\draw [bend right=90, looseness=2.00] (62.center) to (61.center);
		\draw [bend right=90, looseness=2.00] (64.center) to (63.center);
	\end{pgfonlayer}
\end{tikzpicture} \]
where the above follows the fact that $(i|i)=1$. This, in conjunction with the previous results, gives:
\[\begin{tikzpicture}
	\begin{pgfonlayer}{nodelayer}
		\node [style=none] (0) at (6.75, 2.25) {};
		\node [style=none] (1) at (7.25, -0) {};
		\node [style=none] (2) at (5.499999, 1.75) {$=$};
		\node [style=cocpoint] (3) at (8.75, 2.25) {$i$};
		\node [style=none] (4) at (7.25, 0.9999999) {};
		\node [style=none] (5) at (7.25, 0.9999999) {};
		\node [style=none] (6) at (7.25, 1.5) {};
		\node [style=none] (7) at (8.25, 1.5) {};
		\node [style=none] (8) at (8.25, -0.4999999) {};
		\node [style=none] (9) at (7.25, -0.4999999) {};
		\node [style=none] (10) at (7.25, -0) {};
		\node [style=none] (11) at (7.25, 0.9999999) {};
		\node [style=none] (12) at (8.25, 0.9999999) {};
		\node [style=none] (13) at (8.25, -0) {};
		\node [style=none] (14) at (9.000001, 0.9999999) {};
		\node [style=none] (15) at (9.000001, -0) {};
		\node [style=none] (16) at (7.75, 0.4999999) {$T'_i$};
		\node [style=none] (17) at (0, 2) {};
		\node [style=none] (18) at (2.5, -0.9999999) {};
		\node [style=none] (19) at (2.5, 2) {};
		\node [style=none] (20) at (0.9999999, 2.5) {};
		\node [style=none] (21) at (0.9999999, 0.7499999) {};
		\node [style=none] (22) at (0.9999999, -0.4999999) {};
		\node [style=none] (23) at (2.5, 0.7499999) {};
		\node [style=none] (24) at (2.5, -0.4999999) {};
		\node [style=none] (25) at (0.9999999, 2) {};
		\node [style=none] (26) at (2.5, 2.5) {};
		\node [style=none] (27) at (0.9999999, -0.9999999) {};
		\node [style=cocpoint] (28) at (3.5, 2) {$i$};
		\node [style=none] (29) at (2.5, 2) {};
		\node [style=none] (30) at (3.5, 0.7499999) {};
		\node [style=none] (31) at (3.5, -0.4999999) {};
		\node [style=none] (32) at (1.75, 0.7499999) {$T$};
		\node [style=none] (33) at (12.25, -0) {};
		\node [style=none] (34) at (14, 0.9999999) {};
		\node [style=none] (35) at (12.25, -0) {};
		\node [style=none] (36) at (14, -0) {};
		\node [style=cocpoint] (37) at (13.75, 2.25) {$i$};
		\node [style=none] (38) at (12.25, 0.9999999) {};
		\node [style={small box}] (39) at (13, 0.9999999) {$T_i$};
		\node [style=none] (40) at (11.75, 2.25) {};
		\node [style=none] (41) at (12.25, 0.9999999) {};
		\node [style=none] (42) at (13.25, -0) {};
		\node [style=none] (43) at (12.25, 0.9999999) {};
		\node [style=none] (44) at (10.5, 1.75) {$=$};
		\node [style=none] (45) at (14.75, -0) {};
	\end{pgfonlayer}
	\begin{pgfonlayer}{edgelayer}
		\draw [bend right=90, looseness=2.00] (4.center) to (1.center);
		\draw (0.center) to (3);
		\draw (4.center) to (5.center);
		\draw (6.center) to (9.center);
		\draw (9.center) to (8.center);
		\draw (8.center) to (7.center);
		\draw (7.center) to (6.center);
		\draw (12.center) to (14.center);
		\draw (13.center) to (15.center);
		\draw (20.center) to (27.center);
		\draw (27.center) to (18.center);
		\draw (18.center) to (26.center);
		\draw (26.center) to (20.center);
		\draw (17.center) to (25.center);
		\draw [bend right=90, looseness=2.00] (21.center) to (22.center);
		\draw (29.center) to (28);
		\draw (23.center) to (30.center);
		\draw (24.center) to (31.center);
		\draw [bend right=90, looseness=2.00] (38.center) to (33.center);
		\draw (40.center) to (37);
		\draw (38.center) to (43.center);
		\draw (42.center) to (36.center);
		\draw (38.center) to (39);
		\draw (39) to (34.center);
		\draw (33.center) to (42.center);
	\end{pgfonlayer}
\end{tikzpicture} \]
\endproof

\begin{lemma}$\exists\ T'$ such that,
\[\begin{tikzpicture}
	\begin{pgfonlayer}{nodelayer}
		\node [style=none] (0) at (6.75, 0.7499999) {};
		\node [style=none] (1) at (1.5, 0.7499999) {$T$};
		\node [style=none] (2) at (2, -0.7499999) {};
		\node [style=none] (3) at (0.9999999, 1.75) {};
		\node [style=none] (4) at (4.5, 0.7499999) {$=$};
		\node [style=none] (5) at (0.9999999, 0.7499999) {};
		\node [style=none] (6) at (2.750001, -0.25) {};
		\node [style=none] (7) at (0.9999999, 2.25) {};
		\node [style=none] (8) at (0.9999999, -0.25) {};
		\node [style=none] (9) at (0.7499999, -0.25) {};
		\node [style=none] (10) at (2.750001, 0.7499999) {};
		\node [style=none] (11) at (2.750001, -0.25) {};
		\node [style=none] (12) at (2, -0.25) {};
		\node [style=none] (13) at (8.000001, 0.7499999) {};
		\node [style=none] (14) at (8.75, -0.25) {};
		\node [style=none] (15) at (0.7499999, 0.7499999) {};
		\node [style=none] (16) at (8.25, -0.25) {};
		\node [style=none] (17) at (2, 2.25) {};
		\node [style=none] (18) at (2, 1.75) {};
		\node [style=none] (19) at (2, 0.7499999) {};
		\node [style=none] (20) at (6.75, -0.25) {};
		\node [style=none] (21) at (0.9999999, -0.7499999) {};
		\node [style=cpoint] (22) at (0, 1.75) {$\sigma$};
		\node [style=none] (23) at (2.750001, 1.75) {};
		\node [style=none] (24) at (6.999999, 0.25) {};
		\node [style=none] (25) at (8.000001, 0.25) {};
		\node [style=cpoint] (26) at (5.999999, 1.75) {$\sigma$};
		\node [style=none] (27) at (8.000001, 1.75) {};
		\node [style=none] (28) at (6.999999, 2.25) {};
		\node [style=none] (29) at (8.000001, 2.25) {};
		\node [style=none] (30) at (8.75, 1.75) {};
		\node [style=none] (31) at (8.75, 0.7499999) {};
		\node [style=none] (32) at (6.999999, 0.7499999) {};
		\node [style=none] (33) at (6.999999, 1.75) {};
		\node [style=none] (34) at (7.499999, 1.25) {$T'$};
		\node [style=none] (35) at (10.5, -0) {$\forall |\sigma)$};
		\node [style=none] (36) at (11.25, -0) {};
	\end{pgfonlayer}
	\begin{pgfonlayer}{edgelayer}
		\draw (7.center) to (17.center);
		\draw (17.center) to (2.center);
		\draw (2.center) to (21.center);
		\draw [bend right=90, looseness=2.25] (15.center) to (9.center);
		\draw (7.center) to (21.center);
		\draw (15.center) to (5.center);
		\draw (9.center) to (8.center);
		\draw (12.center) to (11.center);
		\draw (19.center) to (10.center);
		\draw [bend right=90, looseness=2.25] (0.center) to (20.center);
		\draw (20.center) to (14.center);
		\draw (28.center) to (24.center);
		\draw (24.center) to (25.center);
		\draw (25.center) to (29.center);
		\draw (29.center) to (28.center);
		\draw (26) to (33.center);
		\draw (0.center) to (32.center);
		\draw (13.center) to (31.center);
		\draw (27.center) to (30.center);
		\draw (22) to (3.center);
		\draw (18.center) to (23.center);
	\end{pgfonlayer}
\end{tikzpicture} \]\
\end{lemma}
\proof
\[\begin{tikzpicture}
	\begin{pgfonlayer}{nodelayer}
		\node [style=none] (0) at (10.75, 1.25) {};
		\node [style=none] (1) at (1.5, 11.25) {$T$};
		\node [style=none] (2) at (2, 9.75) {};
		\node [style=none] (3) at (1, 12.25) {};
		\node [style=none] (4) at (4.5, 11) {$=$};
		\node [style=none] (5) at (1, 11.25) {};
		\node [style=trace] (6) at (2.75, 12.25) {};
		\node [style=none] (7) at (1, 12.75) {};
		\node [style=none] (8) at (1, 10.25) {};
		\node [style=none] (9) at (0.5000001, 10.25) {};
		\node [style=trace] (10) at (2.75, 11.25) {};
		\node [style=none] (11) at (2, 12.25) {};
		\node [style=none] (12) at (12.25, 1.25) {};
		\node [style=none] (13) at (13, 0.2500001) {};
		\node [style=none] (14) at (0.5000001, 11.25) {};
		\node [style=none] (15) at (12.5, 0.2500001) {};
		\node [style=none] (16) at (2, 12.75) {};
		\node [style=none] (17) at (2, 10.25) {};
		\node [style=none] (18) at (2, 11.25) {};
		\node [style=none] (19) at (10.75, 0.2500001) {};
		\node [style=none] (20) at (1, 9.75) {};
		\node [style=none] (21) at (0.5000001, 12.25) {};
		\node [style=none] (22) at (2.75, 10.25) {};
		\node [style=none] (23) at (11.25, 0.7500002) {};
		\node [style=none] (24) at (12.25, 0.7500002) {};
		\node [style=none] (25) at (10.75, 2.25) {};
		\node [style=none] (26) at (12.25, 2.25) {};
		\node [style=none] (27) at (11.25, 2.75) {};
		\node [style=none] (28) at (12.25, 2.75) {};
		\node [style=none] (29) at (13, 2.25) {};
		\node [style=none] (30) at (13, 1.25) {};
		\node [style=none] (31) at (11.25, 1.25) {};
		\node [style=none] (32) at (11.25, 2.25) {};
		\node [style=none] (33) at (11.75, 1.75) {$T'$};
		\node [style=cocpoint] (34) at (15.5, 12.75) {$i$};
		\node [style=none] (35) at (15.75, 10.5) {};
		\node [style=none] (36) at (16, 9.5) {};
		\node [style=none] (37) at (16.25, 10.5) {};
		\node [style=trace] (38) at (16, 11.5) {};
		\node [style=none] (39) at (14, 11.5) {};
		\node [style=none] (40) at (14, 10.5) {};
		\node [style={small box}] (41) at (14.75, 11.5) {$T_i$};
		\node [style=none] (42) at (0.5000001, 9.25) {};
		\node [style=none] (43) at (0.5000001, 12.25) {};
		\node [style=none] (44) at (2.75, 9.25) {};
		\node [style=none] (45) at (14, 9.5) {};
		\node [style=none] (46) at (14, 12.75) {};
		\node [style=none] (47) at (8.5, 11.25) {$T$};
		\node [style=none] (48) at (8.999999, 13) {};
		\node [style=none] (49) at (7.999999, 13) {};
		\node [style=trace] (50) at (9.75, 11.25) {};
		\node [style=none] (51) at (8.999999, 12.5) {};
		\node [style=none] (52) at (9.75, 10.25) {};
		\node [style=none] (53) at (8.999999, 11.25) {};
		\node [style=cocpoint] (54) at (9.75, 12.5) {$i$};
		\node [style=none] (55) at (9.75, 9.25) {};
		\node [style=none] (56) at (8.999999, 10.25) {};
		\node [style=none] (57) at (7.999999, 12.5) {};
		\node [style=none] (58) at (7.5, 11.25) {};
		\node [style=none] (59) at (7.5, 12.5) {};
		\node [style=none] (60) at (7.999999, 11.25) {};
		\node [style=none] (61) at (7.5, 10.25) {};
		\node [style=none] (62) at (8.999999, 9.75) {};
		\node [style=none] (63) at (7.999999, 10.25) {};
		\node [style=none] (64) at (7.999999, 9.75) {};
		\node [style=none] (65) at (7.5, 12.5) {};
		\node [style=none] (66) at (7.5, 9.25) {};
		\node [style=none] (67) at (5.499999, 11) {$\sum_i$};
		\node [style=none] (68) at (12, 11) {$\sum_i$};
		\node [style=none] (69) at (11, 11) {$=$};
		\node [style=none] (70) at (13, 6.75) {};
		\node [style=none] (71) at (13, 5.75) {};
		\node [style=trace] (72) at (14, 6.75) {};
		\node [style=none] (73) at (14, 5.75) {};
		\node [style=none] (74) at (11, 6.75) {$=$};
		\node [style=none] (75) at (3, 3.75) {Purification $\implies$};
		\node [style=none] (76) at (5.75, 2.25) {};
		\node [style=none] (77) at (5.75, 1.25) {};
		\node [style=none] (78) at (4.75, 2.75) {};
		\node [style=none] (79) at (6.5, 2.25) {};
		\node [style=none] (80) at (5.75, 0.2500001) {};
		\node [style=none] (81) at (4.75, 1.25) {};
		\node [style=none] (82) at (4.75, -0.2499996) {};
		\node [style=none] (83) at (4.25, 1.25) {};
		\node [style=none] (84) at (5.75, -0.2499996) {};
		\node [style=none] (85) at (4.75, 0.2500001) {};
		\node [style=none] (86) at (4.25, 0.2500001) {};
		\node [style=none] (87) at (6.5, -0.7499997) {};
		\node [style=none] (88) at (6.5, 0.2500001) {};
		\node [style=none] (89) at (5.25, 1.25) {$T$};
		\node [style=none] (90) at (4.25, 2.25) {};
		\node [style=none] (91) at (4.25, 2.25) {};
		\node [style=none] (92) at (4.25, -0.7499997) {};
		\node [style=none] (93) at (6.5, 1.25) {};
		\node [style=none] (94) at (4.75, 2.25) {};
		\node [style=none] (95) at (5.75, 2.75) {};
		\node [style=none] (96) at (8.000001, 1.25) {$=$};
		\node [style=none] (97) at (10.75, 2.25) {};
		\node [style=none] (98) at (10.75, -0.7499997) {};
		\node [style=none] (99) at (10.75, 2.25) {};
		\node [style=none] (100) at (13, -0.7499997) {};
		\node [style=none] (101) at (13.75, 0.5000001) {};
		\node [style=none] (102) at (13, 7.75) {};
		\node [style=none] (103) at (14, 4.75) {};
		\node [style=none] (104) at (13, 4.75) {};
		\node [style=trace] (105) at (14, 7.75) {};
	\end{pgfonlayer}
	\begin{pgfonlayer}{edgelayer}
		\draw (7.center) to (16.center);
		\draw (16.center) to (2.center);
		\draw (2.center) to (20.center);
		\draw [bend right=90, looseness=2.25] (14.center) to (9.center);
		\draw (7.center) to (20.center);
		\draw (14.center) to (5.center);
		\draw (9.center) to (8.center);
		\draw (18.center) to (10);
		\draw [bend right=90, looseness=2.25] (0.center) to (19.center);
		\draw (19.center) to (13.center);
		\draw (27.center) to (23.center);
		\draw (23.center) to (24.center);
		\draw (24.center) to (28.center);
		\draw (28.center) to (27.center);
		\draw (25.center) to (32.center);
		\draw (0.center) to (31.center);
		\draw (12.center) to (30.center);
		\draw (26.center) to (29.center);
		\draw (21.center) to (3.center);
		\draw (17.center) to (22.center);
		\draw (11.center) to (6);
		\draw [bend right=90, looseness=2.25] (39.center) to (40.center);
		\draw (40.center) to (37.center);
		\draw (39.center) to (41);
		\draw (41) to (38);
		\draw [bend left=90, looseness=1.75] (42.center) to (43.center);
		\draw (42.center) to (44.center);
		\draw [bend left=90, looseness=1.50] (45.center) to (46.center);
		\draw (46.center) to (34);
		\draw (45.center) to (36.center);
		\draw (49.center) to (48.center);
		\draw (48.center) to (62.center);
		\draw (62.center) to (64.center);
		\draw [bend right=90, looseness=1.75] (58.center) to (61.center);
		\draw (49.center) to (64.center);
		\draw (58.center) to (60.center);
		\draw (61.center) to (63.center);
		\draw (53.center) to (50);
		\draw (59.center) to (57.center);
		\draw (56.center) to (52.center);
		\draw (51.center) to (54);
		\draw [bend left=90, looseness=1.25] (66.center) to (65.center);
		\draw (66.center) to (55.center);
		\draw [bend right=90, looseness=2.25] (70.center) to (71.center);
		\draw (70.center) to (72);
		\draw (71.center) to (73.center);
		\draw (78.center) to (95.center);
		\draw (95.center) to (84.center);
		\draw (84.center) to (82.center);
		\draw [bend right=90, looseness=2.25] (83.center) to (86.center);
		\draw (78.center) to (82.center);
		\draw (83.center) to (81.center);
		\draw (86.center) to (85.center);
		\draw (77.center) to (93.center);
		\draw (91.center) to (94.center);
		\draw (80.center) to (88.center);
		\draw (76.center) to (79.center);
		\draw [bend left=90, looseness=1.75] (92.center) to (90.center);
		\draw (92.center) to (87.center);
		\draw [bend left=90, looseness=1.75] (98.center) to (99.center);
		\draw (98.center) to (100.center);
		\draw [bend right=90, looseness=1.50] (102.center) to (104.center);
		\draw (102.center) to (105);
		\draw (104.center) to (103.center);
	\end{pgfonlayer}
\end{tikzpicture}
\]
Dynamic faithfulness then gives the result.
\endproof
\begin{theorem} $T'$ is a controlled transformation, $T'=C\{T_i\}$.\label{Existance}
\end{theorem}
\proof
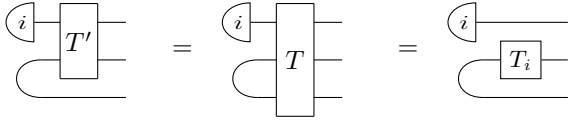
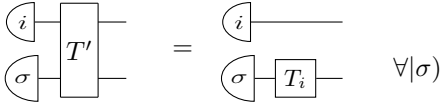
\[\begin{tikzpicture}
	\begin{pgfonlayer}{nodelayer}
		\node [style=none] (0) at (0.5000004, 5.75) {};
		\node [style=none] (1) at (7.249999, 5.75) {$T$};
		\node [style=none] (2) at (7.749999, 4.25) {};
		\node [style=none] (3) at (6.750001, 6.75) {};
		\node [style=none] (4) at (6.750001, 5.75) {};
		\node [style=none] (5) at (6.750001, 7.25) {};
		\node [style=none] (6) at (6.750001, 4.75) {};
		\node [style=none] (7) at (6.25, 4.75) {};
		\node [style=none] (8) at (7.749999, 6.75) {};
		\node [style=none] (9) at (2, 5.75) {};
		\node [style=none] (10) at (2.75, 4.75) {};
		\node [style=none] (11) at (6.25, 5.75) {};
		\node [style=none] (12) at (2.25, 4.75) {};
		\node [style=none] (13) at (7.749999, 7.25) {};
		\node [style=none] (14) at (7.749999, 4.75) {};
		\node [style=none] (15) at (7.749999, 5.75) {};
		\node [style=none] (16) at (0.5000004, 4.75) {};
		\node [style=none] (17) at (6.750001, 4.25) {};
		\node [style=cpoint] (18) at (5.75, 6.75) {$i$};
		\node [style=none] (19) at (8.500001, 4.75) {};
		\node [style=none] (20) at (0.9999995, 5.25) {};
		\node [style=none] (21) at (2, 5.25) {};
		\node [style=cpoint] (22) at (0, 6.75) {$i$};
		\node [style=none] (23) at (2, 6.75) {};
		\node [style=none] (24) at (0.9999995, 7.25) {};
		\node [style=none] (25) at (2, 7.25) {};
		\node [style=none] (26) at (2.75, 6.75) {};
		\node [style=none] (27) at (2.75, 5.75) {};
		\node [style=none] (28) at (0.9999995, 5.75) {};
		\node [style=none] (29) at (0.9999995, 6.75) {};
		\node [style=none] (30) at (1.5, 6.25) {$T'$};
		\node [style=none] (31) at (4.249999, 6) {$=$};
		\node [style=none] (32) at (8.500001, 6.75) {};
		\node [style=none] (33) at (8.500001, 5.75) {};
		\node [style=none] (34) at (10.25, 6) {$=$};
		\node [style=none] (35) at (14.5, 4.75) {};
		\node [style=cpoint] (36) at (11.75, 6.75) {$i$};
		\node [style=none] (37) at (14.5, 6.75) {};
		\node [style=none] (38) at (12.25, 4.75) {};
		\node [style=none] (39) at (14.5, 5.75) {};
		\node [style=none] (40) at (12.25, 5.75) {};
		\node [style={small box}] (41) at (13.25, 5.75) {$T_i$};
		\node [style=none] (42) at (5.000001, 2.75) {Dynamically faithful state $\implies$};
		\node [style=none] (43) at (2.25, 1.25) {};
		\node [style=none] (44) at (1.25, -0.7500005) {};
		\node [style=none] (45) at (1.25, 1.25) {};
		\node [style=none] (46) at (1.25, -0.2500002) {};
		\node [style=none] (47) at (2.25, 1.75) {};
		\node [style=none] (48) at (1.25, 1.75) {};
		\node [style=none] (49) at (1.75, 0.5000004) {$T'$};
		\node [style=none] (50) at (2.25, -0.7500005) {};
		\node [style=none] (51) at (3, -0.2500002) {};
		\node [style=none] (52) at (2.25, -0.2500002) {};
		\node [style=none] (53) at (3, 1.25) {};
		\node [style=cpoint] (54) at (0.2500002, -0.2500002) {$\sigma$};
		\node [style=cpoint] (55) at (0.2500002, 1.25) {$i$};
		\node [style=none] (56) at (8.75, -0.2500002) {};
		\node [style={small box}] (57) at (7.5, -0.2500002) {$T_i$};
		\node [style=none] (58) at (8.75, 1.25) {};
		\node [style=cpoint] (59) at (5.999999, -0.2500002) {$\sigma$};
		\node [style=cpoint] (60) at (5.999999, 1.25) {$i$};
		\node [style=none] (61) at (4.5, 0.5000004) {$=$};
		\node [style=none] (62) at (10.75, -0) {$\forall |\sigma)$};
	\end{pgfonlayer}
	\begin{pgfonlayer}{edgelayer}
		\draw (5.center) to (13.center);
		\draw (13.center) to (2.center);
		\draw (2.center) to (17.center);
		\draw [bend right=90, looseness=2.25] (11.center) to (7.center);
		\draw (5.center) to (17.center);
		\draw (11.center) to (4.center);
		\draw (7.center) to (6.center);
		\draw [bend right=90, looseness=2.25] (0.center) to (16.center);
		\draw (16.center) to (10.center);
		\draw (24.center) to (20.center);
		\draw (20.center) to (21.center);
		\draw (21.center) to (25.center);
		\draw (25.center) to (24.center);
		\draw (22) to (29.center);
		\draw (0.center) to (28.center);
		\draw (9.center) to (27.center);
		\draw (23.center) to (26.center);
		\draw (18) to (3.center);
		\draw (14.center) to (19.center);
		\draw (8.center) to (32.center);
		\draw (15.center) to (33.center);
		\draw [bend right=90, looseness=2.25] (40.center) to (38.center);
		\draw (36) to (37.center);
		\draw (38.center) to (35.center);
		\draw (40.center) to (41);
		\draw (41) to (39.center);
		\draw (48.center) to (44.center);
		\draw (44.center) to (50.center);
		\draw (50.center) to (47.center);
		\draw (47.center) to (48.center);
		\draw (55) to (45.center);
		\draw (54) to (46.center);
		\draw (52.center) to (51.center);
		\draw (43.center) to (53.center);
		\draw (60) to (58.center);
		\draw (59) to (57);
		\draw (57) to (56.center);
	\end{pgfonlayer}
\end{tikzpicture} \] which is the defining characteristic of $C\{T_i\}$.
\endproof

\section{Superposition preservation\label{SuperpositionProof}}
Lem.~\ref{SP} already gives some notion of superposition preservation, we can use our other results above to extend this.
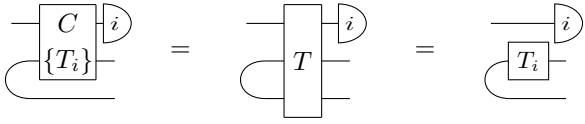
\proof\[\begin{tikzpicture}
	\begin{pgfonlayer}{nodelayer}
		\node [style=none] (0) at (0.7500001, 0.5000003) {};
		\node [style=none] (1) at (7.999999, 0.5000003) {$T$};
		\node [style=none] (2) at (8.5, -1.000001) {};
		\node [style=none] (3) at (7.5, 1.5) {};
		\node [style=none] (4) at (7.5, 0.5000003) {};
		\node [style=none] (5) at (7.5, 2) {};
		\node [style=none] (6) at (7.5, -0.5000003) {};
		\node [style=none] (7) at (7, -0.5000003) {};
		\node [style=none] (8) at (8.5, 1.5) {};
		\node [style=none] (9) at (2.5, 0.5000003) {};
		\node [style=none] (10) at (3, -0.5000003) {};
		\node [style=none] (11) at (7, 0.5000003) {};
		\node [style=none] (12) at (2.5, -0.5000003) {};
		\node [style=none] (13) at (8.5, 2) {};
		\node [style=none] (14) at (8.5, -0.5000003) {};
		\node [style=none] (15) at (8.5, 0.5000003) {};
		\node [style=none] (16) at (0.7500001, -0.5000003) {};
		\node [style=none] (17) at (7.5, -1.000001) {};
		\node [style=none] (18) at (9.25, -0.5000003) {};
		\node [style=none] (19) at (1.000001, -0) {};
		\node [style=none] (20) at (2.5, -0) {};
		\node [style=none] (21) at (2.5, 1.5) {};
		\node [style=none] (22) at (1.000001, 2) {};
		\node [style=none] (23) at (2.5, 2) {};
		\node [style=none] (24) at (3, 0.5000003) {};
		\node [style=none] (25) at (1.000001, 0.5000003) {};
		\node [style=none] (26) at (1.000001, 1.5) {};
		\node [style=none] (27) at (1.749999, 1.5) {$C$};
		\node [style=none] (28) at (4.75, 0.7500001) {$=$};
		\node [style=none] (29) at (9.25, 0.5000003) {};
		\node [style=none] (30) at (11.25, 0.7500001) {$=$};
		\node [style=none] (31) at (15, -0.5000003) {};
		\node [style=none] (32) at (13.5, -0.5000003) {};
		\node [style=none] (33) at (15, 0.5000003) {};
		\node [style=none] (34) at (13.5, 0.5000003) {};
		\node [style={small box}] (35) at (14, 0.5000003) {$T_i$};
		\node [style=none] (36) at (0.2500005, 1.5) {};
		\node [style=none] (37) at (6.5, 1.5) {};
		\node [style=none] (38) at (13, 1.5) {};
		\node [style=none] (39) at (1.749999, 0.5000003) {$\{T_i\}$};
		\node [style=cocpoint] (40) at (3, 1.5) {$i$};
		\node [style=cocpoint] (41) at (15, 1.5) {$i$};
		\node [style=cocpoint] (42) at (9.25, 1.5) {$i$};
		\node [style=none] (43) at (15.5, -0) {};
	\end{pgfonlayer}
	\begin{pgfonlayer}{edgelayer}
		\draw (5.center) to (13.center);
		\draw (13.center) to (2.center);
		\draw (2.center) to (17.center);
		\draw [bend right=90, looseness=2.25] (11.center) to (7.center);
		\draw (5.center) to (17.center);
		\draw (11.center) to (4.center);
		\draw (7.center) to (6.center);
		\draw [bend right=90, looseness=2.25] (0.center) to (16.center);
		\draw (16.center) to (10.center);
		\draw (22.center) to (19.center);
		\draw (19.center) to (20.center);
		\draw (20.center) to (23.center);
		\draw (23.center) to (22.center);
		\draw (0.center) to (25.center);
		\draw (9.center) to (24.center);
		\draw (14.center) to (18.center);
		\draw (15.center) to (29.center);
		\draw [bend right=90, looseness=2.25] (34.center) to (32.center);
		\draw (32.center) to (31.center);
		\draw (34.center) to (35);
		\draw (35) to (33.center);
		\draw (36.center) to (26.center);
		\draw (37.center) to (3.center);
		\draw (21.center) to (40);
		\draw (8.center) to (42);
		\draw (38.center) to (41);
	\end{pgfonlayer}
\end{tikzpicture} \]
The first equality uses Thm.~\ref{Existance} and the second Lem.~\ref{SP}. Using the result of dynamically faithful states then implies,
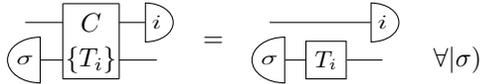
\[\begin{tikzpicture}
	\begin{pgfonlayer}{nodelayer}
		\node [style=none] (0) at (2.75, 0.9999995) {};
		\node [style=none] (1) at (1.25, -0.5000004) {};
		\node [style=none] (2) at (1.25, 0.9999995) {};
		\node [style=none] (3) at (1.25, -0) {};
		\node [style=none] (4) at (2.75, 1.5) {};
		\node [style=none] (5) at (1.25, 1.5) {};
		\node [style=none] (6) at (2, 0.9999995) {$C$};
		\node [style=none] (7) at (2.75, -0.5000004) {};
		\node [style=none] (8) at (3.749999, -0) {};
		\node [style=none] (9) at (2.75, -0) {};
		\node [style=cpoint] (10) at (0.2500002, -0) {$\sigma$};
		\node [style=none] (11) at (9.750001, -0) {};
		\node [style={small box}] (12) at (8.25, -0) {$T_i$};
		\node [style=cpoint] (13) at (6.750001, -0) {$\sigma$};
		\node [style=none] (14) at (5.25, 0.5000004) {$=$};
		\node [style=none] (15) at (0.2500002, 0.9999995) {};
		\node [style=none] (16) at (6.750001, 0.9999995) {};
		\node [style=none] (17) at (2, -0) {$\{T_i\}$};
		\node [style=cocpoint] (18) at (3.749999, 0.9999995) {$i$};
		\node [style=cocpoint] (19) at (9.750001, 0.9999995) {$i$};
		\node [style=none] (20) at (11.75, -0) {$\forall |\sigma)$};
	\end{pgfonlayer}
	\begin{pgfonlayer}{edgelayer}
		\draw (5.center) to (1.center);
		\draw (1.center) to (7.center);
		\draw (7.center) to (4.center);
		\draw (4.center) to (5.center);
		\draw (10) to (3.center);
		\draw (9.center) to (8.center);
		\draw (13) to (12);
		\draw (12) to (11.center);
		\draw (15.center) to (2.center);
		\draw (0.center) to (18);
		\draw (16.center) to (19);
	\end{pgfonlayer}
\end{tikzpicture} \]\endproof

This actually only proves the existence of a controlled transformation that preserves superpositions, it is simple to show that it must be true for all controlled transformations using an argument analogous to Lem.~\ref{SP}.

\section{Definition of higher-order interference} \label{Higher-order-App}

The original definition of higher-order interference was in the framework of quantum measure theory \cite{sorkin1994quantum, sorkin1995quantum}. The definition revolves around the concepts of `histories' and the `quantum measure'. Histories correspond to paths through space-time, a set of histories $A$ is any collection of these paths. We will be concerned with sets of histories with some initial condition $s$ and some final condition $e$ along with an intermediate condition $i$ that `the history passes through slit $i$'. We label these sets of histories $A^{se}_i$. In an interference experiment it is common to have some way to create a path difference between the different slits, either by introducing some `phase shifter' or by moving the final detection point, we label this data $t$.

The quantum measure $\mu$ associates some probability to each set of histories, which should be thought of as the probability that a particle `has a history from that set'. In general $\mu$ will depend on the experimental control $t$.

The existence of higher-order interference in this framework is as follows:

\begin{definition}
$n$-th order interference $\iff$ $\exists s,e$ s.t.\[\mu\left[\bigcup_{i\in\mathds{P}}A^{se}_i\right](t) \neq \sum_{I\subset \mathds{P}}(-1)^{n-|I|+1}\mu\left[\bigcup_{i\in I}A^{se}_{i}\right](t)\]
\end{definition}

Given this definition we can provide a translation to the operational definition that we are using,

\vspace{0.5cm}

\begin{tabular}{|l|c|c|}
  \hline
     & QMT & GPT \\\hline
  Initial condition & $s$ & $|s)$ \\
  Final condition & $e$ & $(e|$ \\
  Experimental control & $t$ & $T$ \\
  Probability of path $i$ & $\mu[A^{se}_i](t)$ & $\mathcal{C}_{se_{\{i\}}}(T)$  \\
  Probability of subset $I$ & $\mu[\bigcup_{i\in I}A^{se}_i](t)$ & $\mathcal{C}_{se_{I}}(T)$ \\
  \hline
\end{tabular}

\vspace{.5cm}

Note that some ambiguity is introduced in switching to the operational framework, which $e_I\in\mathcal{E}_I$ should be picked? Other approaches to defining higher-order interference in operational theories \cite{Higher-order-reconstruction} have required sufficient structure to define a set of `filters', $\{F_I\}$, for the theory. These are transformations that represent the action of leaving open some subset of slits $I$ whilst closing the others, in which case $(e_I|=(e|F_{I}$. However, arbitrary theories do not have sufficient structure to define filters and so one must consider all possible choices $(e_I|$ with the correct support. This leads to the following definition of $n$th-order interference,

\begin{definition}
$n$-th order interference $\iff$ $\exists s,e$ s.t.\[\mathcal{C}_{s,e}(T)\neq \sum_{I\subset\mathds{P}}(-1)^{n-|I|+1}\mathcal{C}_{s,e_I}(T),\]
$\forall (e_I|\in \mathcal{E}_I$.
\end{definition}
The introduction of the `$\forall$' statement compared to the original definition is due to the ambiguity in choosing which effect corresponds to blocking some subset of paths. In the main text the explicit dependence on $T$ in the above equation will be suppressed, as $\mathcal{C}_{se}$ has already been defined as a function from the phase group to probabilities.

For example, the existence of second-order interference implies that there exists $|s)$ and $(e|$ such that
\[\mathcal{C}_{s,e}\neq \mathcal{C}_{s,e_{\{0\}}}+\mathcal{C}_{s,e_{\{1\}}},\] $\forall\ |e_{\{i\}})\in\mathcal{E}_i$. Whilst the existence of third-order interference corresponds to the existence of,
$|s)$ and $(e|$ such that
\[\mathcal{C}_{s,e}\neq \mathcal{C}_{s,e_{\{0,1\}}}+\mathcal{C}_{s,e_{\{1,2\}}}+\mathcal{C}_{s,e_{\{2,0\}}}-\mathcal{C}_{s,e_{\{0\}}}-\mathcal{C}_{s,e_{\{1\}}}-\mathcal{C}_{s,e_{\{2\}}},\] $\forall\ |e_{I})\in\mathcal{E}_I$.

We consider the above for the case of quantum theory to provide some intuition for the definitions. Firstly, we show the existence of second-order interference. Define our paths by $p_i:=(\ket{i}\bra{i},\ket{i}\bra{i})$, then choose $|s)=\ket{+}\bra{+}=(e|$. Then $(e_{\{i\}}|\in \{r_i \ket{i}\bra{i}\}$ where $r_i$ is an arbitrary positive real number. The phase group is given by $\mathcal{P}:=\{e^{i\theta_0}\ket{0}\bra{0}+e^{i\theta_1}\ket{1}\bra{1}\}$. It is then simple to show that,
\[\mathcal{C}_{s,e}(T) = \cos^2\left(\frac{\theta_0-\theta_1}{2}\right),\]
whilst,
\[\mathcal{C}_{s,e_{\{0\}}}(T)+\mathcal{C}_{s,e_{\{1\}}}(T)= \frac{r_0+r_1}{\sqrt{2}}.\]
It is then simple to see that, as functions of $\theta_i$,
\[\cos^2\left(\frac{\theta_0-\theta_1}{2}\right)\neq \frac{r_0+r_1}{\sqrt{2}},\]
for any choice of $r_i$, i.e. $(e_{\{i\}}|$. Therefore -- by our definition -- quantum theory has second-order interference as we would expect.

Next we consider our definition of third-order interference for quantum theory. We consider a specific choice of $|s)$ and $(e|$, and note that this can be simply -- but tediously -- generalised to all choices. Consider $|s)=\frac{1}{3}(\ket{0}+\ket{1}+\ket{2})(\bra{0}+\bra{1}+\bra{2})=(e|$, and the phase group, $\mathcal{P}=\{e^{i\theta_0}\ket{0}\bra{0}+e^{i\theta_1}\ket{1}\bra{1}+e^{i\theta_2}\ket{2}\bra{2}\}$. Then let $(e_{\{i,j\}}|=\frac{1}{3}(\ket{i}+\ket{j})(\bra{i}+\bra{j})$ and $(e_{\{i\}}|=\frac{1}{3}\ket{i}\bra{i}$. Note that these are sub-normalised effects. It is then simple to check our definition for this particular choice of $|s)$ and $(e|$. The interference patterns can be written as,
\begin{enumerate}
\item $$\begin{aligned}
\mathcal{C}_{s,e}(T) &= \frac{1}{9}|e^{i\theta_0}+e^{i\theta_1}+e^{i\theta_2}|^2 \\ &=\frac{1}{9}\left(3+\sum_{i>j}e^{i(\theta_i-\theta_j)}+e^{i(\theta_j-\theta_i)}\right), \end{aligned} $$

\item $$ \begin{aligned} \mathcal{C}_{s,e_{\{i,j\}}}(T) &= \frac{1}{9}|e^{i\theta_i}+e^{i\theta_j}|^2 \\ &=\frac{1}{9}\left(2+e^{i(\theta_i-\theta_j)}+e^{i(\theta_j-\theta_i)}\right),\end{aligned}$$
\item $$\begin{aligned} \mathcal{C}_{s,e_{\{i\}}}(T) = \frac{1}{9}|e^{i\theta_i}|^2=\frac{1}{9} \end{aligned}$$
\end{enumerate}
and so,
\[\mathcal{C}_{s,e}(T)=\sum_{i>j}\mathcal{C}_{s,e_{\{i,j\}}}(T)-\sum_i \mathcal{C}_{s,e_{\{i\}}}(T).\]
This proves that the particular choice of state $|s)$ and effect $(e|$ do not give higher-order interference for quantum theory. This can, however, be readily generalised to hold for any choice, and so demonstrates that quantum theory does not exhibit higher-order interference.


\begin{thebibliography}{10}

\bibitem{arkhipov}
S. Aaronson and A. Arkhipov, {\em The Computational Complexity of Linear Optics}.
\newblock Proc. of the Forty-third Annual ACM Symposium on Theory of Computing (STOC 2011), pp. 333-342, 2011.

\bibitem{shor}
P. Shor, {\em Polynomial-Time Algorithms for Prime Factorization and Discrete Logarithms on a Quantum Computer}.
\newblock SIAM J. Sci. Statist. Comput. 26, 1484, 1997.

\bibitem{kickback}
R. Cleve, A. Ekert, C. Macchiavello and M. Mosca, {\em Quantum algorithms revisited},
\newblock Proceedings of the Royal Society of London A: Mathematical, Physical and Engineering Sciences, volume 454, 1998.

\bibitem{Nielsen}
M.~A. Nielsen and I.~L. Chuang, {\em Quantum computation and Quantum
  information}.
\newblock Cambridge University press, 2000.

\bibitem{Pavia1}
G.~Chiribella, G.~M. D'Ariano, and P.~Perinotti, {\em Probabilistic theories
  with purification}.
\newblock Phys. Rev. A 81, 062348, 2010.

\bibitem{Pavia2}
G.~Chiribella, G.~M. D'Ariano, and P.~Perinotti, {\em Informational derivation of Quantum Theory}.
\newblock Phys. Rev. A 84, 012311, 2011.

\bibitem{Barrett-2007}
J.~Barrett, {\em Information processing in generalized probabilistic theories}.
\newblock Phys. Rev. A 75 No. 3, 032304, 2007.

\bibitem{Hardy-2011}
L.~Hardy, {\em Reformulating and reconstructing quantum theory}.
\newblock arXiv:quant-ph/1104.2066v3, 2011.

\bibitem{LB-2014}
C.~M.~Lee and J.~Barrett, {\em Computation in generalised probabilistic theories}.
\newblock arXiv:quant-ph/1412.8671 , 2014.

\bibitem{PR-van-Dam}
W.~van Dam, {\em Implausible consequences of superstrong nonlocality}.
\newblock arXiv:quant-ph/0501159, 2005.

\bibitem{Higher-order-reconstruction}
H. Barnum, M. P. Mueller and C. Ududec, {\em Higher-order interference and single system postulates for quantum theory}.
\newblock New Journal of Physics 16, 123029, 2014.

\bibitem{Niestegge-2012}
G. Niestegge, {\em Conditional probability, three-slit experiments and the Jordan structure of quantum mechanics}.
\newblock Advances in Mathematical Physics 156573, 2012.

\bibitem{Henson-2015}
J. Henson, {\em Bounding quantum contextuality with lack of third-order interference}.
\newblock Phys. Rev. Lett. 114, 220403, 2015.

\bibitem{ududec2011three}
C.~Ududec, H.~Barnum and J.~Emerson,  {\em Three slit experiments and the structure of quantum theory}.
\newblock Foundations of Physics, 41, 3, pages 396--405, 2011

\bibitem{sinha2008testing}
U. Sinha, C. Couteau, Z. Medendorp, I. S{\"o}llner, R. Laflamme and R. Sorkin, {\em Testing Born's rule in quantum mechanics with a triple slit experiment}.
\newblock arXiv:quant-ph/0811.2068, 2008.

\bibitem{park2012three}
D. Park, O. Moussa and R. Laflamme, {\em Three path interference using nuclear magnetic resonance: a test of the consistency of Born's rule}.
\newblock New Journal of Physics, 14, 11, 113025, 2012.

\bibitem{sorkin1994quantum}
R.~Sorkin, {\em Quantum mechanics as quantum measure theory}.
\newblock Modern Physics Letters A, 9, 33, 3119--3127, 1994.

\bibitem{sorkin1995quantum}
R.~Sorkin, {\em Quantum measure theory and its interpretation}.
\newblock arXiv:gr-qc/9507057, 1995.

\bibitem{Interference-speed-up}
D. Stahlke, {\em Quantum interference as a resource for quantum speedup}.
\newblock 	Phys. Rev. A 90, 022302, 2014.

\bibitem{Garner}
A.~Garner, O.~Dahlsten, Y. Nakata, M. Murao and V. Vedral, {\em Phase and interference in generalised probabilistic theories}.
\newblock New J. Phys. 15 093044 2013.

\bibitem{Rev}
D.~Gross, M.~Mueller, R.~Colbeck, and O.~Dahlsten, {\em All reversible dynamics
  in maximal non-local theories are trivial}.
\newblock Phys. Rev. Lett. 104, 080402, 2010.

\bibitem{MU}
M.~Mueller and C.~Ududec, {\em The structure of reversible computation determines the self-duality of quantum theory}.
\newblock Phy. Rev. lett. 108, 130401, 2012.

\bibitem{Oscar}
O. Dahlsten, A. Garner, J. Thompson, M. Gu and V. Vedral, {\em Particle exchange in post-quantum theories}.
\newblock arXiv:quant-ph/1307.2529-v2, 2013.


\bibitem{DensityCube}
B.~Daki{\'c}, T.~Paterek and {\v{C}}.~Brukner, {\em Density cubes and higher-order interference theories}.
\newblock New Journal of Physics, Vol. 16 , 2014.

\bibitem{Thesis}
C.~Ududec, {\em Perspectives on the formalism of quantum theory}.
\newblock University of Waterloo, 2012.

\bibitem{LS-2015}
C. M. Lee and J. H. Selby, {\em Higher-order interference in extensions of quantum theory}.
\newblock arXiv:quant-ph/1510.03860, 2015.

\bibitem{torons}
G. N. Afanasiev, {\em Quantum mechanics of toroidal anyons}.
\newblock J. Phys. A: Math. Gen. 24 2517, 1991.

\bibitem{Barton}
B. Zwiebach, {\em A first course in String Theory}.
\newblock Cambridge University press, 2009.

\bibitem{work}
K. Korzekwa, M. Lostaglio, J. Oppenheim and D. Jennings, {\em The extraction of work from quantum coherence}.
 \newblock arXiv:quant-ph/1506.07875, 2015.

\bibitem{cqm2}
 B. Coecke, {\em Quantum picturalism}.
 \newblock Contemporary physics, 51, 1, 2010.

\bibitem{cqm1}
S. Abramsky and B. Coecke, {\em A categorical semantics of quantum protocols}.
 \newblock Proceedings of the 19th Annual IEEE Symposium on Logic in Computer Science, 2004.

 \bibitem{Thermo}
H. Barnum, J. Barrett, M. Krumm and M. Mueller, {\em Entropy, majorization and thermodynamics in general probabilistic theories}.
\newblock  arXiv:quant-ph/1508.03107, 2015.

\bibitem{Thermo1}
G. Chiribella and C.M. Scandolo, {\em Entanglement and thermodynamics in general probabilistic theories}.
\newblock arXiv:quant-ph/1504.07045, 2015.

\bibitem{Thermo2}
G. Chiribella and C.M. Scandolo, {\em Operational axioms for state diagonalization}.
\newblock arXiv:quant-ph/1506.00380,2015.

\bibitem{barnum2015some}
H. Barnum, M. Graydon and A. Wilce, {\em Some Nearly Quantum Theories}.
\newblock arXiv:1507.06278, 2015.

\bibitem{Proofs}
C. M. Lee and M. J. Hoban, {\em Proofs and advice in general physical theories:a trade-off between states and dynamics?},
\newblock arXiv:1510.04702, 2015.

\bibitem{landscape}
J. Barrett, N. de Beaudrap, M. J. Hoban and C. M. Lee, {\em the computational landscape of general physical theories,}
\newblock Forthcoming, 2016.

\end{thebibliography}
\end{document}